\documentclass[unsortedaddress,superscriptaddress,pra,notitlepage]{revtex4-1}

\usepackage{amsmath,amssymb}
\usepackage{braket}
\usepackage{theorem}
\usepackage{color}

\newtheorem{definition}{Definition}[section]
\newtheorem{theorem}[definition]{Theorem}
\newtheorem{prop}[definition]{Proposition}
\newtheorem{lemma}[definition]{Lemma}
\newtheorem{remark}[definition]{Remark}
\newtheorem{corollary}[definition]{Corollary}

\def\a{\alpha}
\def\R{\mathbb{R}}

\def\E{\mathcal{E}}
\def\F{\mathcal{F}}
\def\H{\mathcal{H}}
\def\K{\mathcal{K}}
\def\L{\mathcal{L}}
\def\M{\mathcal{M}}
\def\S{\mathcal{S}}
\DeclareMathOperator{\Tr}{Tr}

\def\hil{{\mathcal H}}
\def\kil{{\mathcal K}}

\def\B{{\mathcal B}}
\def\I{\mathcal{I}}

\def\M{\mathcal{M}}

\def\S{{\mathcal S}}
\def\X{{\mathcal X}}

\def\half{\frac{1}{2}}
\def\iff{\Longleftrightarrow}
\def\imp{\Longrightarrow}
\def\ep{\varepsilon}
\def\bN{\mathbb{N}}
\def\bC{\mathbb{C}}
\def\bR{\mathbb{R}}

\def\bz{\left(}
\def\jz{\right)}

\def\kii{\emph}
\def\kiii{}

\def\map{\F}
\def\mapp{\F'}
\def\old{^{\mathrm{(old)}}}
\def\nw{^{\mathrm{(new)}}}
\def\what{\widehat}

\newcommand{\ki}{\emph}

\newcommand{\ds}{\mbox{ }\mbox{ }}

\newcommand{\vecc}[1]{\underline{#1}}
\newcommand{\diad}[2]{|#1\rangle\langle #2|}
\newcommand{\pr}[1]{\diad{#1}{#1}}

\newcommand{\sr}[2]{D\bz #1\,\|\, #2\jz}
\newcommand{\rsr}[3]{D_{#3}\bz #1\,\|\, #2\jz}
\newcommand{\rsro}[3]{D_{#3}^{\mathrm{(old)}}\bz #1\,\|\, #2\jz}
\newcommand{\rsrn}[3]{D_{#3}^{\mathrm{(new)}}\bz #1\,\|\, #2\jz}
\newcommand{\dmax}[2]{D_{\max}\bz #1\,\|\, #2\jz}

\DeclareMathOperator{\supp}{supp}

\DeclareMathOperator{\Exp}{\mathbb{E}}

\def\nn{\nonumber \\}

\newenvironment{proof}[1][Proof:]{\begin{trivlist}
\item[\hskip \labelsep {\it #1}]}{\hfill$\Box$\end{trivlist}}

\begin{document}

\title{
Quantum hypothesis testing and the operational \\
interpretation of the quantum R\'enyi relative entropies
}

\author{Mil\'an Mosonyi}
\email{milan.mosonyi@gmail.com}
\affiliation{
F\'{\i}sica Te\`{o}rica: Informaci\'{o} i Fenomens Qu\`{a}ntics,
Universitat Aut\`{o}noma de Barcelona, ES-08193 Bellaterra (Barcelona), Spain.
}
\affiliation{
Mathematical Institute, Budapest University of Technology and Economics, \\
Egry J\'ozsef u~1., Budapest, 1111 Hungary.
}

\author{Tomohiro Ogawa}
\email{ogawa@is.uec.ac.jp}
\affiliation{
Graduate School of Information Systems,
University of Electro-Communications,
1-5-1 Chofugaoka, Chofu-shi, Tokyo, 182-8585, Japan.
}

\begin{abstract}
We show that the new quantum extension of R\'enyi's $\alpha$-relative entropies, introduced recently by
M\"uller-Lennert, Dupuis, Szehr, Fehr and Tomamichel, \kiii{J.~Math.~Phys.} \textbf{54}, 122203, (2013), and Wilde, Winter, Yang, \kiii{Commun.~Math.~Phys.}, \textbf{331}, (2014),
have an operational interpretation in the strong converse problem of quantum hypothesis testing.
Together with related results for the direct part of quantum hypothesis testing, known as the quantum Hoeffding bound, our result suggests that the operationally relevant definition of the quantum R\'enyi relative entropies depends on the parameter $\alpha$: for $\alpha<1$, the right choice seems to be the traditional definition $\rsro{\rho}{\sigma}{\alpha}:=\frac{1}{\alpha-1}\log\Tr\rho^{\alpha}\sigma^{1-\alpha}$, whereas for $\alpha>1$ the right choice is the newly introduced version
$\rsrn{\rho}{\sigma}{\alpha}:=
\frac{1}{\alpha-1}\log\Tr\bz\sigma^{\frac{1-\alpha}{2\alpha}}\rho\sigma^{\frac{1-\alpha}{2\alpha}}\jz^{\alpha}$.

As a sideresult, we
show that the new R\'enyi $\alpha$-relative entropies are asymptotically attainable by measurements for $\alpha>1$, and
give a new simple proof for their monotonicity under completely positive trace-preserving maps.
\end{abstract}

\maketitle

\section{Introduction}

R\'enyi in his seminal paper \cite{Renyi} introduced a generalization of the Kullback-Leibler divergence (relative entropy). According to his definition, the $\alpha$-divergence of two probability distributions (more generally, two positive functions) $p$ and $q$ on a finite set $\X$ for a parameter $\alpha\in[0,+\infty)\setminus\{1\}$ is given by
\begin{align}\label{Renyi def}
\rsr{p}{q}{\alpha}:=
\begin{cases}
\frac{1}{\alpha-1}\log\sum_{x\in\X}p(x)^{\alpha}q(x)^{1-\alpha}
-\frac{1}{\alpha-1}\log\sum_{x\in\X}p(x),
& \supp p\subseteq\supp q\text{ or }\alpha\in[0,1),\\
+\infty,&\text{otherwise}.
\end{cases}
\end{align}
The limit $\alpha\to 1$ yields the standard relative entropy.
These quantities turned out to play a central role in information theory and statistics; indeed, the R\'enyi relative entropies and derived quantities quantify the trade-off between the exponents of the relevant quantities in many information-theoretic tasks, including hypothesis testing, source coding and noisy channel coding; see, e.g.~\cite{Csiszar} for an overview of these results. It was also shown in \cite{Csiszar} that the R\'enyi relative entropies, and other related quantities, like the R\'enyi entropies and the R\'enyi capacities, have direct operational interpretations as so-called generalized cutoff rates in the corresponding information-theoretic tasks.

In quantum theory, the state of a system is described by a density operator instead of a probability distribution, and the definition \eqref{Renyi def} can be extended for pairs of density operators (more generally, positive operators) in various inequivalent ways, due to the non-commutativity of operators.
There are some basic requirements any such extension should satisfy; most importantly, positivity and monotonicity under CPTP (completely positive and trace-preserving) maps. That is, if $D_{\alpha}$ is an extension of \eqref{Renyi def} to pairs of positive semidefinite operators, then it should satisfy
\begin{align*}
\rsr{\rho}{\sigma}{\alpha}\ge 0\ds\ds\ds\text{and}\ds\ds\ds\rsr{\rho}{\sigma}{\alpha}= 0\iff\rho=\sigma
& &\text{(positivity)}
\end{align*}
for any density operators $\rho,\sigma$ and $\alpha>0$, and if $\map$ is a CPTP map then
\begin{align}
\rsr{\map(\rho)}{\map(\sigma)}{\alpha}\le\rsr{\rho}{\sigma}{\alpha}& &\text{(monotonicity)}
\label{eq:2}
\end{align}
should hold.

One formal extension has been known in the literature for a long time, defined as
\begin{equation}\label{old Renyi}
\rsro{\rho}{\sigma}{\alpha}:=
\begin{cases}
\frac{1}{\alpha-1}\log\Tr\rho^{\alpha}\sigma^{1-\alpha}-\frac{1}{\alpha-1}\log\Tr\rho,&
\supp\rho\subseteq\supp\sigma\text{ or }\alpha\in[0,1),\\
+\infty,&\text{otherwise}.
\end{cases}
\end{equation}
H\"older's inequality ensures positivity of $D_{\alpha}^{\mathrm{(old)}}$ for every
$\alpha>0$. Monotonicity has been proved for $\alpha\in[0,2]\setminus\{1\}$ with various
methods \cite{Lieb,Petz,Uhlmann}, but it doesn't hold for $\alpha>2$ in general, as it was noted, e.g., in \cite{Renyi_new}. Monotonicity under measurements, however, is still true for $\alpha>2$ \cite{H:text}.
In the limit $\alpha\to 1$, these divergences yield Umegaki's
relative entropy \cite{Umegaki}
\begin{equation}\label{Umegaki}
\rsr{\rho}{\sigma}{1}:=\lim_{\alpha\to 1}\rsro{\rho}{\sigma}{\alpha}=
\sr{\rho}{\sigma}:=
\begin{cases}
\frac{1}{\Tr\rho}\Tr\rho(\log\rho-\log\sigma),&\supp\rho\subseteq\supp\sigma,\\
+\infty,&\text{otherwise}.
\end{cases}
\end{equation}
The quantum Stein's lemma \cite{HP,ON} gives an operational interpretation to Umegaki's relative entropy (which we will call simply relative entropy for the rest)
in a state discrimination problem,
as the optimal decay rate of the type II error under the assumption that the type I error goes to $0$ (see section \ref{sec:hypotesting} for details).
This shows that Umegaki's relative entropy is the right non-commutative extension of the Kullback-Leibler divergence from an information-theoretic point of view.

It has been shown in \cite{MH} that, similarly to the classical case, the R\'enyi $\alpha$-relative entropies
$D_{\alpha}^{\mathrm{(old)}}$ with $\alpha\in(0,1)$ have a direct operational interpretation as generalized cutoff rates
in binary state discrimination. This in turn is based on the so-called quantum Hoeffding bound theorem, that
quantifies the trade-off between the optimal exponential decay rates of the two error probabilities in binary state discrimination \cite{Hayashi,HMO2,Nagaoka,ANSzV}. In more detail, it says that if the type II error is required to vanish asymptotically as $\sim e^{-nr}$ for some $r>0$
($n$ is the number of the copies of the system, all prepared in state $\rho$ or all prepared in state $\sigma$)
then the optimal type I error goes to $0$ exponentially fast with the exponent given by the Hoeffding divergence
\begin{equation}\label{Hoeffding div}
H_r(\rho\|\sigma):=\sup_{0<\alpha<1}\frac{\alpha-1}{\alpha}\left[r-\rsro{\rho}{\sigma}{\alpha}\right],
\end{equation}
as long as
$r<\sr{\rho}{\sigma}$. The transformation rule defining $H_r(\rho\|\sigma)$ from the $\alpha$-relative entropies can be inverted, and $\rsro{\rho}{\sigma}{\alpha}$ can be expressed in terms of the Hoeffding divergences for any $\alpha\in(0,1)$. These results suggest that $D_{\alpha}^{\mathrm{(old)}}$ gives the right quantum extension of the R\'enyi $\alpha$-relative entropies for the parameter range $\alpha\in(0,1)$.

Recently, a new quantum extension of the R\'enyi $\alpha$-relative entropies have been proposed in \cite{Renyi_new,WWY}, defined as
\begin{equation}\label{new Renyi def}
\rsrn{\rho}{\sigma}{\alpha}:=\begin{cases}
\frac{1}{\alpha-1}\log\Tr\bz\sigma^{\frac{1-\alpha}{2\alpha}}\rho\sigma^{\frac{1-\alpha}{2\alpha}}\jz^{\alpha}-\frac{1}{\alpha-1}\log\Tr\rho,&
\supp\rho\subseteq\supp\sigma\text{ or }\alpha\in[0,1),\\
+\infty,&\text{otherwise}.
\end{cases}
\end{equation}
These new R\'enyi divergences also yield Umegaki's relative entropy in the limit $\alpha\to 1$.
Monotonicity for the range $\alpha\in(1,2]$ has been shown in \cite{Renyi_new,WWY} and extended to $\alpha\in(1,+\infty)$ in \cite{Beigi} and, independently and with a different proof method, for the range $\alpha\in[\half,1)\cup(1,+\infty)$ in \cite{FL}. It is claimed in \cite{Renyi_new} that these new R\'enyi relative entropies are not monotone for $\alpha\in[0,\half)$. Positivity follows immediately from
the monotonicity for $\alpha\in[\half,1)\cup(1,+\infty)$. The Araki-Lieb-Thirring inequality
\cite{Araki,LT} (see also \cite[Theorem IX.2.10]{Bhatia}) implies that
\begin{equation}\label{ALT}
\rsrn{\rho}{\sigma}{\alpha}\le\rsro{\rho}{\sigma}{\alpha}
\end{equation}
for every $\rho,\sigma$ and $\alpha\in(0,+\infty)\setminus\{1\}$.
Moreover, the results of \cite{Hiai ALT} yield that for non-commuting operators the above inequality is strict for all
$\alpha\in(0,+\infty)\setminus\{1\}$.
 The converse Araki-Lieb-Thirring inequality of \cite{Aud-ALT} implies lower bounds on $D_{\alpha}\nw$ in terms of $D_{\alpha}\old$ \cite{M}.

In this paper we show that the new R\'enyi relative entropies with $\alpha>1$ play the same role in the converse part of binary state discrimination as the old R\'enyi relative entropies with $\alpha\in(0,1)$ play in the direct part. Namely, we show
(in Theorem \ref{thm:exponent}) that
if the type II error is required to vanish asymptotically as $\sim e^{-nr}$ with some $r>\sr{\rho}{\sigma}$
then the optimal type I error goes to $1$ exponentially fast,
with the exponent given by the converse Hoeffding divergence
\begin{equation}\label{converse Hoeffding}
H_r^*(\rho\|\sigma):=\sup_{1<\alpha}\frac{\alpha-1}{\alpha}\left[r-\rsrn{\rho}{\sigma}{\alpha}
\right].
\end{equation}
From this,
we derive (in Theorem \ref{thm:cutoff}) a representation of the new R\'enyi relative
entropies as generalized cutoff rates in the strong converse domain, thus providing a direct
operational interpretation of
the new R\'enyi relative entropies for $\alpha>1$.
These results are direct quantum counterparts of the well-known classical results by Han and Kobayashi \cite{HK} and Csisz\'ar \cite{Csiszar}.
In the quantum case, Hayashi \cite{H:text} obtained a limiting formula
for the strong converse exponent
using the classical R\'enyi relative entropies;
see Remarks \ref{Remark-Hayashi-1} and \ref{Remark-Hayashi-2}.
Our formula \eqref{converse Hoeffding} can be seen 
as a single-letterization of Hayashi's exponent.

In the proof we only use the monotonicity of the new R\'enyi relative entropies under pinching \cite[Proposition 13]{Renyi_new}, and show (in Theorem \ref{thm:attainability}) that the new R\'enyi relative entropies can be asymptotically attained by measurements, similarly to the relative entropy \cite{HP}. Based on this,
we provide a simple new proof for the monotonicity of $D_{\alpha}^{\mathrm{(new)}}$ under CPTP maps for $\alpha>1$ as a side-
result. We give an overview of the monotonicity and attainability properties of the old and the new R\'enyi relative 
entropies in Appendix \ref{sec:mon}.

Our results suggest that, somewhat surprisingly, the right formula to define the R\'enyi $\alpha$-relative entropies for quantum states depends on whether the parameter $\alpha$ is below or above $1$; it seems that for $\alpha<1$, one should use the old R\'enyi relative entropies, while for $\alpha>1$, the new R\'enyi relative entropies are the right choice. Hence, we suggest to define the R\'enyi relative entropies for quantum states
(more generally, for positive operators) $\rho,\sigma$
as
\begin{equation*}
\rsr{\rho}{\sigma}{\alpha}:=\begin{cases}
\frac{1}{\alpha-1}\log\Tr\rho^{\alpha}\sigma^{1-\alpha}-\frac{1}{\alpha-1}\log\Tr\rho,&
\alpha\in[0,1),\\
\frac{1}{\alpha-1}\log\Tr\bz\sigma^{\frac{1-\alpha}{2\alpha}}\rho\sigma^{\frac{1-\alpha}{2\alpha}}\jz^{\alpha}
-\frac{1}{\alpha-1}\log\Tr\rho, & \alpha>1\text{ and }\supp\rho\subseteq\supp\sigma,\\
+\infty,&\text{otherwise}.
\end{cases}
\end{equation*}

\section{Preliminaries}

For a finite-dimensional Hilbert space $\hil$, let $\L(\hil)$ denote the set of linear operators on $\hil$,
let $\L(\hil)_+$ denote the set of positive semidefinite operators,
and $\S(\H)$ be the set of density operators (states) on $\hil$ (i.e., positive semidefinite operators with trace $1$).
A finite-valued POVM (positive operator valued measure) on $\hil$ is a map $M:\,\I\to\L(\hil)$, where $\I$ is some finite set, $0\le M_i,\,i\in\I$, and $\sum_{i\in\I}M_i=I$. We denote the set of POVMs on $\hil$ by
$\M(\hil)$.

Any Hermitian operator $A\in\L(\hil)$ admits a spectral decomposition $A=\sum_i a_iP_i$, where $a_i\in\bR$ and the $P_i$ are orthogonal projections. We introduce the notation
$\{A>0\}:=\sum_{i:\,a_i>0}P_i$ for the spectral projection of $A$
corresponding to the positive half-line $(0,+\infty)$.
The spectral projections $\{A\ge 0\},\,\{A<0\}$ and $\{A\le 0\}$
are defined similarly.
The positive part of $A$ is defined as
\begin{align}
A_+:= A \{A>0\},
\label{eq:37}
\end{align}
and it is easy to see that
\begin{align}
\Tr A_+=\Tr A \{A>0\} = \max_{0\le T\le I} \Tr AT\ge 0.
\label{eq:38}
\end{align}
In particular, if $\rho_n$ and $\sigma_n$ are self-adjoint operators then for any $a\in\bR$ the
application of \eqref{eq:38} to $A=\rho_n-e^{na}\sigma_n$ yields
\begin{align}
\Tr\rho_n\{\rho_n-e^{na}\sigma_n>0\}
\ge e^{na}\Tr\sigma_n\{\rho_n-e^{na}\sigma_n>0\}.
\label{eq:39}
\end{align}

If $\F$ is a positive trace-preserving map then
\begin{align*}
\Tr\F(A)_+&=\max_{0\le T\le I}\Tr\F(A)T=\max_{0\le T\le I}\Tr A\F^*(T)
\le\max_{0\le S\le I}\Tr AS=\Tr A_+.
\end{align*}
In particular, we have the following lemma.
\begin{lemma}
\label{mono:positive}
Let $\rho_n$ and $\sigma_n$ be self-adjoint operators
and $\F$ be a positive trace-preserving map.
Then for any $a\in\bR$,
\begin{align}
\Tr(\rho_n-e^{na}\sigma_n)_+
\ge\Tr(\F(\rho_n)-e^{na}\F(\sigma_n))_+ .
\label{eq:40}
\end{align}
\end{lemma}
\medskip

Let $A$ be a Hermitian operator on $\hil$ with spectral decomposition $A=\sum_i a_iE_i$. The \ki{pinching operation} $\E_A$ corresponding to $A$ is defined as
\begin{align}
\E_A(B):=\sum_iE_iBE_i,\ds\ds\ds B\in\L(\hil).
\label{eq:8}
\end{align}
It is also denoted by $\E_E(B)$ in terms of the PVM (projection-valued measure) $E=\{E_i\}_i$.
Note that $\E_A(B)$ is the unique operator in the commutant $\{A\}'$ of $\{A\}$ satisfying
\begin{align}
\forall C\in\{A\}',\ds\Tr BC = \Tr \E_A(B)C.
\label{eq:9}
\end{align}

The following lemma is from \cite{H:pinching,H:text}:
\begin{lemma}[pinching inequality]\label{lemma:pinching}
Let $A$ be self-adjoint and $B$ be a positive semidefinite operator on $\hil$. Then
\begin{equation*}
B\le v(A)\E_A(B),
\end{equation*}
where $v(A)$ denotes the number of different eigenvalues of $A$.
\end{lemma}

All through the paper, $\rho$ and $\sigma$ will denote positive semidefinite operators on some finite-dimensional Hilbert space $\hil$, and we use the notation
\begin{equation}\label{eq:10}
\rho_n:=\rho^{\otimes n},\ds\ds\ds
\sigma_n:=\sigma^{\otimes n},\ds\ds\ds
\widehat\rho_n:=\E_{\sigma_{n}}(\rho_n),\ds\ds\ds
v_n:=v(\sigma_n),
\end{equation}
where $\E_{\sigma_{n}}$ is the pinching operation corresponding to $\sigma_n$, and $v_n$ denotes the number of different eigenvalues of $\sigma_n$.
Note that $v_n\le(n+1)^{\dim\H}$, and lemma \ref{lemma:pinching} yields
\begin{align}
\rho_n\le v_n\widehat\rho_n\le (n+1)^{\dim\H}\widehat\rho_n.
\label{eq:11}
\end{align}
The power of the pinching inequality for asymptotic analysis comes from the fact that
\begin{equation*}
\lim_{n\to+\infty}\frac{1}{n}\log v_n=0,
\end{equation*}
which we will use repeatedly and without further explanation in the paper.
\smallskip

We will use the convention that powers of a positive semidefinite operator are only taken on its support and defined to be $0$ on the orthocomplement of its support.
That is, if $a_1,\ldots,a_r$ are the eigenvalues of $A\ge 0$, with corresponding eigenprojections $P_1,\ldots,P_r$, then $A^{p}:=\sum_{i:\,a_i>0}a_i^p P_i$ for any
$p\in\bR$. In particular, $A^0$ is the projection onto the support of $A$.
We will also use the convention $\log 0:=-\infty$.

\section{Properties of the new R\'enyi relative entropies}\label{sec:attainability}

For positive semidefinite operators $\rho$ and $\sigma$, and $\alpha\in\bR$, let
\begin{align}
F_{\a}(\rho\|\sigma):=\log
\Tr\left(\sigma^{\frac{1-\a}{2\a}}\rho\sigma^{\frac{1-\a}{2\a}}\right)^{\a}.
\label{eq:5}
\end{align}
For a POVM $M=\{M_x\}_x$,
we can consider the corresonding classical quantity as
\begin{align}
F_{\a}^{M}(\rho\|\sigma)
:=\log\left(\sum_x\{\Tr\rho M_x\}^{\a}\{\Tr\sigma M_x\}^{1-\a}\right).
\label{eq:6}
\end{align}
Note that for states $\rho$ and $\sigma$ such that $\supp\rho\subseteq\supp\sigma$, $\frac{1}{\alpha-1}F_{\a}(\rho\|\sigma)$ is the new R\'enyi $\alpha$-relative entropy defined in \eqref{new Renyi def}, and $\frac{1}{\alpha-1}F_{\a}^{M}(\rho\|\sigma)$ is the post-measurement R\'enyi $\alpha$-relative entropy.

In this section we show that for every $\alpha>1$, the new R\'enyi $\alpha$-relative entropies are asymptotically attainable by measurements in the limit of infinitely many copies of $\rho$ and $\sigma$; for this we only use that the new R\'enyi $\alpha$-relative entropies are monotonic under pinching by the reference state, which is very simple to show.
From this we derive a new simple proof for the monotonicity of the new R\'enyi $\alpha$-relative entropies.

Monotonicity in the classical case is well-known and easy to prove; we state it explicitly here for completeness:
\begin{lemma}[classical monotonicity]\label{lemma:classical monotonicity}
Let $\rho,\sigma\in\B(\hil)_+$ be commuting operators such that $\supp\rho\subseteq\supp\sigma$, and let $\F:\,\B(\hil)\to\B(\kil)$ be a positive trace-preserving map such that $\F(\rho)$ commutes with $\F(\sigma)$.
For every $\alpha>1$,
$F_{\alpha}(\F(\rho)\|\F(\sigma))\le F_{\alpha}(\rho\|\sigma)$.
\end{lemma}
\begin{proof}
The proof is an elementary argument based on the convexity of the function $x\mapsto x^{\alpha}$ on $[0,+\infty)$ for $\alpha>1$; details can bee found e.g.~in
\cite[Proposition A.3]{HMPB}.
\end{proof}

The following has been shown in \cite[Proposition 13]{Renyi_new}. We reproduce the proof here for readers' convenience.
\begin{lemma}[monotonicity under pinching]
\label{mono:pinching}
Let $\rho,\sigma\in\L(\hil)_+$ and $\alpha\ge 1$. Then
\begin{align}
F_{\a}(\E_{\sigma}(\rho)\|\sigma)\le F_{\a}(\rho\|\sigma).
\label{eq:12}
\end{align}
\end{lemma}
\begin{proof}
It is easy to see that
$\sigma^{\frac{1-\a}{2\a}}\E_{\sigma}(\rho)\sigma^{\frac{1-\a}{2\a}}
=
\E_{\sigma}\bz\sigma^{\frac{1-\a}{2\a}}\rho\sigma^{\frac{1-\a}{2\a}}\jz$,
and Problem II.5.5 with Theorem II.3.1 in \cite{Bhatia}, applied to the convex function $f(t)=t^{\alpha}$, yields the assertion.
\end{proof}

Using the above two lemmas, we can prove monotonicity under measurements.

\begin{lemma}[monotonicity under measurements]
\label{mono:measurement}
Let $\rho,\sigma\in\L(\hil)_+$ be such that $\supp\rho\subseteq\supp\sigma$.
For any POVM $M=\{M_x\}_x\in\M(\hil)$, we have
\begin{align}
F_{\a}^{M}(\rho\|\sigma)\le F_{\a}(\rho\|\sigma), \ds\ds\ds\a\ge 1.
\label{eq:13}
\end{align}
\end{lemma}
\begin{proof}
For any POVM $M_n=\{M_{n}(x)\}_{x}$ on $\hil^{\otimes n}$ and any $\alpha\ge 1$,
\begin{align}
\sum_{x}\bz\Tr\rho_nM_n(x)\jz^{\a}
\bz\Tr\sigma_nM_n(x)\jz^{1-\a}
&\le
v_n^{\a}\sum_{x}\bz\Tr\widehat\rho_nM_n(x)\jz^{\a}
\bz\Tr\sigma_nM_n(x)\jz^{1-\a}\label{eq:14}\\
&\le
v_n^{\a}\Tr\widehat\rho_n^{\a}\sigma_n^{1-\a}
\label{eq:14_1}\\
&\le v_n^{\a}
\Tr\left(\sigma_n^{\frac{1-\a}{2\a}}\rho_n\sigma_n^{\frac{1-\a}{2\a}}\right)^{\a},
\label{eq:16}
\end{align}
where the first inequality is due to \eqref{eq:11}, the second inequality follows from
Lemma \ref{lemma:classical monotonicity}, and the third one from Lemma~\ref{mono:pinching}.

Now let $M=\{M_x\}_{x\in\X}\in\M(\hil)$ be a POVM on a single copy, and
$M_n$ be its $n$th i.i.d.~extension, i.e.,
\begin{align}
M_n(\vecc{x}):=M_{x_1}\otimes\ldots\otimes M_{x_n},\ds\ds\ds
\vecc{x}\in\X^n.
\label{eq:18}
\end{align}
Then we obtain
\begin{align}
\left(
\sum_{x}\bz\Tr\rho M_x\jz^{\a}\bz\Tr\sigma M_x\jz^{1-\a}
\right)^n
&=\sum_{\vecc{x}}\bz\Tr\rho_nM_n(\vecc{x})\jz^{\a}
\bz\Tr\sigma_nM_n(\vecc{x})\jz^{1-\a}
\nn
&\le v_n^{\a}
\Tr\left(\sigma_n^{\frac{1-\a}{2\a}}\rho_n\sigma_n^{\frac{1-\a}{2\a}}\right)^{\a}
\nn
&= v_n^{\a}\bz
\Tr\left(\sigma^{\frac{1-\a}{2\a}}\rho\sigma^{\frac{1-\a}{2\a}}\right)^{\a}
\jz^n.
\label{eq:19}
\end{align}
Taking the logarithm and dividing by $n$ yields
\begin{align}
F_{\a}^{M}(\rho\|\sigma)\le
F_{\a}(\rho\|\sigma)+\frac{\a}{n}\log v_n,
\label{eq:20}
\end{align}
which proves the lemma by taking the limit $n\to\infty$.
\end{proof}

\begin{remark}\label{Remark-Hayashi-1}
The technique used in the proof of the above lemma 
is essentially due to \cite{H:text} (see around page 88),
where the inequalities \eqref{eq:14} and \eqref{eq:14_1}
have been shown.
\end{remark}

\begin{remark}
Note that the assumption $\supp\rho\subseteq\supp\sigma$ was necessary to apply classical monotonicity in \eqref{eq:14_1}.
In fact, the statement of Lemma \ref{mono:measurement} need not hold without this assumption. Indeed, in the extreme case where
$\rho$ and $\sigma$ have orthogonal supports, we have $F_{\alpha}(\rho\|\sigma)=-\infty$, and the trivial POVM $M=\{I\}$ yields
$F_{\alpha}^M(\rho\|\sigma)=\log (\Tr\rho)^{\alpha}(\Tr\sigma)^{1-\alpha}$, which is a finite number unless $\rho$ or $\sigma$ is equal to $0$.
\end{remark}

The following lemma is standard:
\begin{lemma}
\label{lem:norm}
Let $A$ and $B$ be Hermitian operators on $\hil$ with their spectrum in some interval $I$, and let $f:\,I\to\bR$ be a monotone increasing function. If $A\le B$ then $\Tr f(A)\le\Tr f(B)$. In particular,
\begin{equation*}
0\le A\le B\ds\imp\ds \Tr A^{\alpha}\le \Tr B^{\alpha}\ds\ds\ds \alpha>0.
\end{equation*}
\end{lemma}
\begin{proof}
Let $\{\lambda^{\downarrow}_i(A)\}_{i=1}^{\dim\hil}$ denote the sequence of decreasingly ordered eigenvalues of $A$. By the Courant-Fischer-Weyl minimax principle \cite[Corollary III.1.2]{Bhatia}, $\lambda^{\downarrow}_i(A)\le \lambda^{\downarrow}_i(B),\,1\le i\le \dim\hil$, from which the assertion follows.
\end{proof}

\begin{theorem}[asymptotic attainability]
\label{thm:attainability}
Let $\rho,\sigma\in\L(\hil)_+$ be such that $\supp\rho\subseteq\supp\sigma$.
For any $\a\ge 1$, we have
\begin{align}
F_{\a}(\rho\|\sigma)
&=\lim_{n\to\infty}\frac{1}{n}F_{\a}(\widehat\rho_n\|\sigma_n)
\nn
&=\lim_{n\to\infty}\frac{1}{n}\max_{M_n\in\M(\hil^{\otimes n})}F_{\a}^{M_n}(\rho_n\|\sigma_n),
\label{eq:22}
\end{align}
where the maximization in the second line is over all POVMs on $\H^{\otimes n}$.
\end{theorem}

\begin{proof}
Since $\sigma_n$ and $\widehat\rho_n$ commute, they have a common eigenbasis $\{e_n(i)\}_{i=1}^{d_n},\,d_n=(\dim\hil)^n$. Let
$E_n=\{E_n(i)=\pr{e_n(i)}\}_{i=1}^{d_n}$ be the corresponding projection-valued measure. Then
\begin{align}
\frac{1}{n}F_{\a}(\widehat\rho_n\|\sigma_n)
=
\frac{1}{n}F_{\a}^{E_n}(\rho_n\|\sigma_n)
\le
\frac{1}{n}\max_{M_n}F_{\a}^{M_n}(\rho_n\|\sigma_n)
\le
\frac{1}{n}F_{\a}(\rho_n\|\sigma_n)
=
F_{\a}(\rho\|\sigma),
\label{eq:23}
\end{align}
where the last inequality is due to Lemma \ref{mono:measurement}.
By Lemma \ref{lemma:pinching},
\begin{align}
0\le
\sigma_n^{\frac{1-\a}{2\a}}\rho_n\sigma_n^{\frac{1-\a}{2\a}}
\le
v_n
\sigma_n^{\frac{1-\a}{2\a}}\widehat\rho_n\sigma_n^{\frac{1-\a}{2\a}}
=
v_n\sum_{i=1}^{d_n}\bz\Tr\rho_n E_n(i)\jz\bz\Tr\sigma_n E_n(i)\jz^{\frac{1-\a}{\a}}E_n(i),
\label{eq:25}
\end{align}
and Lemma~\ref{lem:norm} yields
\begin{align}
\Tr\left(\sigma_n^{\frac{1-\a}{2\a}}\rho_n\sigma_n^{\frac{1-\a}{2\a}}\right)^{\a}
\le
v_n^{\a}\Tr
\left(\sigma_n^{\frac{1-\a}{2\a}}\widehat\rho_n\sigma_n^{\frac{1-\a}{2\a}}\right)^{\a}
=
v_n^{\a}\sum_{i=1}^{d_n}\bz\Tr\rho_n E_n(i)\jz^{\alpha}\bz\Tr\sigma_n E_n(i)\jz^{1-\a}.
\label{eq:26}
\end{align}
Taking the logarithm, we obtain
\begin{align}
F_{\a}(\rho\|\sigma)\le
\frac{1}{n}F_{\a}(\widehat\rho_n\|\sigma_n)+\frac{\a}{n}\log v_n
=
\frac{1}{n}F_{\a}^{E_n}(\rho_n\|\sigma_n)+\frac{\a}{n}\log v_n
\le
\frac{1}{n}\max_{M_n}F_{\a}^{M_n}(\rho_n\|\sigma_n)+\frac{\a}{n}\log v_n.
\label{eq:27}
\end{align}
Combining this with \eqref{eq:23}, and taking the limit $n\to+\infty$, the assertion follows.
\end{proof}

Theorem \ref{thm:attainability} implies the asymptotic attainability for the R\'enyi relative entropies:
\begin{corollary}\label{cor:Renyi attainability}
For any $\rho,\sigma\in\L(\hil)_+$ and $\a> 1$, we have
\begin{align}
D_{\a}\nw(\rho\|\sigma)
&=\lim_{n\to\infty}\frac{1}{n}D_{\a}\nw(\widehat\rho_n\|\sigma_n)
\nn
&=\lim_{n\to\infty}\frac{1}{n}\max_{M_n\in\M(\hil^{\otimes n})}D_{\a}\nw\bz\{\Tr\rho_n M_n(x)\}_{x\in\X}\|\{\Tr\sigma_n M_n(x)\}_{x\in\X}\jz,\label{Renyi attainability}
\end{align}
where the maximization in the second line is over all POVMs on $\H^{\otimes n}$.
\end{corollary}
\begin{proof}
The case where $\supp\rho\subseteq\supp\sigma$ is immediate from Theorem \ref{thm:attainability}. On the other hand, if
$\supp\rho\nsubseteq\supp\sigma$ then
also $\supp\widehat\rho_n\nsubseteq\supp\sigma_n$, and hence, by the definition \eqref{new Renyi def},
$D_{\a}\nw(\rho\|\sigma)=D_{\a}\nw(\widehat\rho_n\|\sigma_n)=
\max_{M_n\in\M(\hil^{\otimes n})}D_{\a}\nw\bz\{\Tr\rho_n M_n(x)\}_{x\in\X}\|\{\Tr\sigma_n M_n(x)\}_{x\in\X}\jz=+\infty$
for every $n\in\bN$, making the assertion trivial.
\end{proof}

\begin{remark}
The same statement for the relative entropy has been shown in \cite{HP}.
\end{remark}

\begin{remark}
The maximum over all measurements in \eqref{Renyi attainability} can be replaced by a concrete binary POVM given by a Neyman-Pearson test; see Corollary \ref{cor:achievability by NP}.
\end{remark}

Theorem \ref{thm:attainability} has a number of important further corollaries:

\begin{corollary}[convexity]\label{cor:convexity}
For any fixed $\rho,\sigma\in\L(\hil)_+$ such that $\supp\rho\subseteq\supp\sigma$,
$F_{\a}(\rho\|\sigma)$ is a convex function of $\alpha$ for $\a\ge 1$.
\end{corollary}
\begin{proof}
It is easy to see (by computing its second derivative) that $F_{\a}(\widehat\rho_n\|\sigma_n)$ is a convex function of $\alpha$. Thus by Theorem \ref{thm:attainability}, $F_{\a}(\rho\|\sigma)$ is a pointwise limit of convex functions, and hence it is convex.
\end{proof}

\begin{corollary}
For any fixed $\rho,\sigma\in\L(\hil)_+$,
the function
$\alpha\mapsto\rsrn{\rho}{\sigma}{\alpha}$ is monotone increasing for $\alpha>1$.
\end{corollary}
\begin{proof}
We can assume that $\supp\rho\subseteq\supp\sigma$, since otherwise $\rsrn{\rho}{\sigma}{\alpha}=+\infty$ for every $\alpha>1$, and the assertion holds trivially.
Note that $\supp\rho\subseteq\supp\sigma$ implies that $F_{1}(\rho\|\sigma)=\log\Tr\rho$, and hence $\rsrn{\rho}{\sigma}{\alpha}=\frac{F_{\a}(\rho\|\sigma)-F_{1}(\rho\|\sigma)}{\alpha-1}$. The assertion then follows from Corollary \ref{cor:convexity}.
\end{proof}

\begin{corollary}[monotonicity]\label{cor:F monotonicity}
Let $\rho,\sigma\in\L(\hil)_+$ be such that $\supp\rho\subseteq\supp\sigma$, and let $\F:\,\L(\hil)\to\L(\kil)$ be a CPTP map. Then
\begin{equation*}
F_{\a}(\F(\rho)\|\F(\sigma))\le F_{\a}(\rho\|\sigma),\ds\ds\ds \alpha>1.
\end{equation*}
\end{corollary}
\begin{proof}
By complete positivity, $\F_n:=\F^{\otimes n}$ is positive for every $n\in\bN$.
Let $\F_n^*:\L(\K^{\otimes n})\to\L(\H^{\otimes n})$ be the dual (adjoint) of $\F_n$, defined by
\begin{align}
\forall\omega\in\S(\H^{\otimes n}),\,\forall A\in\L(\kil^{\otimes n}),\,
\Tr\F_n(\omega)A=\Tr\omega\F_n^*(A).
\label{eq:29}
\end{align}
Then $\F_n^*$ is a unital positive map.
Thus, if $\{M(x)\}_{x\in\X}\in\M(\kil^{\otimes n})$ is a POVM on $\K^{\otimes n}$ then
$\F_n^*(M):=\{\F_n^*(M(x))\}_{x\in\X}$ is a POVM on $\H^{\otimes n}$.
Hence,
\begin{align}
\max_{M\in\M(\K^{\otimes n})}F_{\a}^{M}(\F_n(\rho_n)\|\F_n(\sigma_n))
=\max_{M\in\M(\K^{\otimes n})}F_{\a}^{\F_n^*(M)}(\rho_n\|\sigma_n)
\le\max_{M\in\M(\H^{\otimes n})}F_{\a}^{M}(\rho_n\|\sigma_n)
\label{eq:30}
\end{align}
for any $n$.
Now \eqref{eq:30} and Theorem~\ref{thm:attainability} yield the assertion.
\end{proof}

Corollary \ref{cor:F monotonicity} immediately implies the following:
\begin{corollary}\label{Renyi monotonicity}
The new R\'enyi relative entropies are monotone under CPTP maps for $\alpha>1$. That is,
if $\rho,\sigma\in\L(\hil)_+$ and $\F:\,\L(\hil)\to\L(\kil)$ is a CPTP map then
\begin{align}
D_{\a}(\F(\rho)\|\F(\sigma))\le D_{\a}(\rho\|\sigma),\ds\ds\ds \alpha>1,
\label{eq:21}
\end{align}
and the limit $\alpha\searrow 1$ yields the same monotonicity property for the relative entropy.
\end{corollary}
\smallskip

For $\rho,\sigma\in\L(\hil)_+$, let
\begin{equation*}
Q_{\alpha}\nw(\rho\|\sigma):=\Tr\left(\sigma^{\frac{1-\a}{2\a}}\rho\sigma^{\frac{1-\a}{2\a}}\right)^{\a},\ds\ds\ds
\alpha\in\bR_+.
\end{equation*}
This is an analogy of the quasi-entropy \cite{Petz} (or quantum $f$-divergence \cite{HMPB}) corresponding to the function $x\mapsto x^{\alpha}$. However, $Q_{\alpha}\nw$ cannot be written in the form of an $f$-divergence \cite[Corollary 2.10]{HMPB}.
Corollary \ref{cor:F monotonicity} is equivalent to the monotonicity of $Q$:
\begin{corollary}[monotonicity of $Q$]\label{cor:Q monotonicity}
Let $\rho,\sigma\in\L(\hil)_+$ be such that $\supp\rho\subseteq\supp\sigma$, and let $\F:\,\L(\hil)\to\L(\kil)$ be a CPTP map. Then
\begin{equation*}
Q_{\a}\nw(\F(\rho)\|\F(\sigma))\le Q_{\a}\nw(\rho\|\sigma),\ds\ds\ds \alpha>1.
\end{equation*}
\end{corollary}

Following the argument of \cite{Petz}, we immediately obtain the joint convexity of $Q$:
\begin{corollary}[joint convexity]\label{cor:joint convexity}
Let $\rho_i,\sigma_i\in\L(\hil)_+$ be such that $\supp\rho_i\subseteq\supp\sigma_i,\,i=1,\ldots,r$, and let
$p_1,\ldots,p_r$ be a probability distribution. Then
\begin{align*}
Q_{\alpha}\nw\bz\sum_{i=1}^rp_i\rho_i\Big\|\sum_{i=1}^rp_i\sigma_i\jz
\le
\sum_{i=1}^r p_i Q_{\alpha}\nw(\rho_i\|\sigma_i).
\end{align*}
\end{corollary}
\begin{proof}
Let $\delta_1,\ldots,\delta_r$ be orthogonal rank $1$ projections on $\kil:=\bC^r$, and define
$\rho:=\sum_{i=1}^rp_i\delta_i\otimes \rho_i$,
$\sigma:=\sum_{i=1}^rp_i\delta_i\otimes \sigma_i$.
Taking $\F:=\Tr_{\kil}$ to be the partial trace over $\kil$ in
Corollary \ref{cor:Q monotonicity}, the assertion follows.
\end{proof}

\begin{remark}
In Corollary \ref{cor:joint convexity}, we obtained the joint convexity from the monotonicity of $Q_{\alpha}\nw$.
In \cite{FL} (and also in \cite{Renyi_new,WWY} for $\alpha\in(1,2]$) the authors followed the opposite approach: they first established joint convexity of $Q_{\alpha}\nw$, and from that they obtained its monotonicity under CPTP maps by a standard argument using the Stinespring representation and decomposing the trace as a convex combination of unitary conjugations.
\end{remark}

\begin{remark}\label{rem:EPPMON}
Note that the monotonicity properties in Corollaries \ref{cor:F monotonicity}, \ref{Renyi monotonicity} and
\ref{cor:Q monotonicity} hold for any trace-preserving linear map $\F$ such that $\F^{\otimes n}$ is positive for every 
$n\in\bN$. This is a weaker condition than complete positivity.
\end{remark}

We give an overview of the various monotonicity and attainability properties of the old and the new R\'enyi relative entropies in Appendix \ref{sec:mon}.

\section{Strong Converse Exponent in Quantum Hypothesis Testing}
\label{sec:sc}

\subsection{Simple Quantum Hypothesis Testing}\label{sec:hypotesting}

We study the simple hypothesis testing problem for the null hypothesis
$H_0$: $\rho_n$ versus the alternative hypothesis $H_1$: $\sigma_n$,
where $\rho_n=\rho^{\otimes n}$ and $\sigma_n=\sigma^{\otimes n}$
are the $n$-fold tensor products of arbitrarily given
density operators $\rho$ and $\sigma$ in $\S(\H)$.
The problem is to decide which hypothesis is true based on the outcome
drawn from a quantum measurement, which is described by a POVM on $\H_n=\hil^{\otimes n}$.
In the hypothesis testing problem,
it is sufficient to treat a two-valued POVM $\{T_n(0),T_n(1)\}\in\M(\hil^{\otimes n})$,
where $0$ and $1$ indicate
the acceptance of $H_0$ and $H_1$, respectively.
Since $T_n(1)=I-T_n(0)$, the POVM is uniquely determined by $T_n=T_n(0)$, and
the only constraint on $T_n$ is that $0\le T_n\le I_n$. We will call such operators tests.
For a test $T_n$, the error probabilities of the first
and the second kind are, respectively, defined by
\begin{align}
\alpha_n(T_n)&:=\Tr\rho_n(I_n-T_n),
\label{eq:31}
\\
\beta_n(T_n)&:=\Tr\sigma_nT_n.
\label{eq:32}
\end{align}

In general there is a trade-off between these error probabilities,
and we can not make these probabilities unconditionally small, as described below.
First, we consider the optimal value for
$\beta_n(T_n)$ under the constant constraint on $\alpha_n(T_n)$, that is,
\begin{align}
\beta_{n}^*(\epsilon):=\min
\Set{\beta_n(T_n) | T_n:\text{test},\,\alpha_n(T_n)\le\epsilon}.
\label{eq:33}
\end{align}
The quantum Stein's lemma \cite{HP,ON} states that
for all $\ep\in(0,1)$,
\begin{align}
\lim_{n\to\infty}\frac{1}{n}\log\beta_n^*(\epsilon) = -D(\rho\|\sigma),
\label{eq:34}
\end{align}
where $D(\rho\|\sigma)$
is the quantum relative entropy given in \eqref{Umegaki}.
This implies the existence of a sequence of tests $\{T_n\}_{n\in\bN}$ such that
\begin{align*}
\lim_{n\to\infty}\frac{1}{n}\log\beta_n(T_n)=-\sr{\rho}{\sigma}
\ds\ds\ds\text{and}\ds\ds\ds
\lim_{n\to\infty}\alpha_n(T_n)=0.
\end{align*}

For the study of the trade-off between the error probabilities, it is natural to ask what
happens if we require the type II error probabilities to vanish with an exponent below or
above the relative entropy, i.e., we want to study
the asymptotic behavior of
$\alpha_n(T_n)$
under the exponential constraint
$\beta_n(T_n)\le e^{-nr},\,r>0$.
Specifically, let us define
\begin{align}
B_e(r)&:=
\sup\left\{-\limsup_{n\to\infty}\frac{1}{n}\log\alpha_n(T_n)\Bigm|
\limsup_{n\to\infty}\frac{1}{n}\log\beta_n(T_n)\le -r\right\}\nonumber
\\
&=
\sup\Bigl\{ R \Bigm|
\exists\{T_n\}_{n=1}^\infty,\,0\le T_n\le I_n,\,\text{s.t.}\nonumber
\\
&\ds\ds\ds\ds\ds\ds\limsup_{n\to\infty}\frac{1}{n}\log\beta_n(T_n)\le -r,\,
\limsup_{n\to\infty}\frac{1}{n}\log\alpha_n(T_n)\le -R
\Bigr\},\label{eq:35}
\end{align}
where the supremum in the first line is taken over all sequences of tests $\{T_n\}_{n\in\bN}$ satisfying the condition.
It was shown in \cite{Hayashi,Nagaoka} that
\begin{align}
B_e(r)=\sup_{0\le s< 1}\frac{-sr-\log\Tr\rho^{1-s}\sigma^{s}}{1-s}=
\sup_{0<\alpha<1}\frac{\alpha-1}{\alpha}\left[r-\rsro{\rho}{\sigma}{\alpha}\right]=
H_r(\rho\|\sigma),
\label{eq:36}
\end{align}
where $D_{\alpha}^{\mathrm{(old)}}$ is the traditional definition of the
quantum R\'enyi relative entropy, given in \eqref{old Renyi},
and $H_r(\rho\|\sigma)$ is the Hoeffding divergence defined in \eqref{Hoeffding div}.
(Note that the roles of the type I and the type II errors are reversed here as compared to some previous work on the Hoeffding bound, and hence our $H_r(\rho\|\sigma)$ corresponds to $H_r(\sigma\|\rho)$ in those works.)
It can be shown that $B_e(r)>0$ when $0<r<D(\rho\|\sigma)$,
and $\alpha_n(T_n)$ goes to zero exponentially
with the rate $B_e(r)$
for an optimal sequence of tests $\{T_n\}_{n=1}^{\infty}$.

On the other hand,
if $\supp\rho\subseteq\supp\sigma$ and $\beta_n(T_n)\le e^{-nr}$ with $r>D(\rho\|\sigma)$ then
$\alpha_n(T_n)$ inevitably goes to 1 exponentially fast \cite{ON};
this is called the strong converse property.
In this case, we are interested in determing the exponent with which
the success probabilities $1-\alpha_n(T_n)=\Tr\rho_nT_n$
go to zero.
The optimal such exponent is the strong converse exponent $B_e^*(r)$; formally,
\begin{align}
B_e^*(r)&:=
\inf\left\{-\liminf_{n\to+\infty}\frac{1}{n}\log\Tr\rho_n T_n\Bigm|
\limsup_{n\to\infty}\frac{1}{n}\log\Tr\sigma_nT_n \le -r\right\},
\label{eq:72}
\end{align}
where the infimum is taken over all possible sequences of tests $\{T_n\}_{n\in\bN}$
satisfying the condition.
Note that one's aim is to make the success probabilities decay as slow as possible, and
hence optimality means taking the smallest possible exponent along all sequences of tests
with a fixed decay rate of the type II errors.
It is easy to see that $B_e^*(r)$ can be alternatively written as
\begin{align}
B_e^*(r)=
\sup\Bigl\{ R \Bigm|
&\forall \{T_n\}_{n=1}^{\infty},\; 0\le T_n \le I_n,
\nn
&\limsup_{n\to\infty}\frac{1}{n}\log\Tr\sigma_nT_n \le -r
\,\Rightarrow\, \liminf_{n\to\infty}\frac{1}{n}\log\Tr\rho_nT_n \le -R
\Bigr\}\nn
=
\inf\Bigl\{ R \Bigm|
&\exists \{T_n\}_{n=1}^{\infty}, \; 0\le T_n \le I_n,
\nn
&\limsup_{n\to\infty}\frac{1}{n}\log\Tr\sigma_nT_n \le -r,\,
\liminf_{n\to\infty}\frac{1}{n}\log\Tr\rho_nT_n \ge -R
\Bigr\} .
\label{sc rate def2}
\end{align}

The main result of Section \ref{sec:sc} is Theorem \ref{thm:exponent}, where we show that, in complete analogy with \eqref{eq:36},
\begin{align}\label{sc exponent equals converse Hoeffding}
B_e^*(r)=\sup_{1<\alpha}\frac{\alpha-1}{\alpha}\left[r-\rsrn{\rho}{\sigma}{\alpha}
\right]=H_r^*(\rho\|\sigma),
\end{align}
where $H_r^*(\rho\|\sigma)$ is the converse Hoeffding divergence \eqref{converse Hoeffding}.
The inequality $B_e^*(r)\ge H_r^*(\rho\|\sigma)$ follows easily from the
monotonicity of the R\'enyi divergences, as we show in Lemma \ref{lemma:converse rate lower
bound}. We show that this is in fact an equality by determining the asymptotics
of the error probabilities for the
Neyman-Pearson tests. This is interesting in itself, as these quantities play a central role in the information spectrum method \cite{Han,NH2007}.
We start with this problem in Section \ref{sec:NP}.

\begin{remark}\label{rem:supports}
Note that if $\supp\rho\subseteq\supp\sigma$ is not satisfied then the strong converse property doesn't hold; indeed, the choice $T_n:=I-\sigma_n^0,\,n\in\bN$, yields a sequence of tests for which $\beta_n(T_n)=0\le e^{-nr},\,r>0$, and
$\alpha_n(T_n)=(\Tr\rho\sigma^0)^n,\,n\in\bN$, which converges to zero exponentially fast with an exponent
$-\log\Tr\rho\sigma^0>0$. Hence, for the rest we will assume that $\supp\rho\subseteq\supp\sigma$.
\end{remark}

\subsection{Exponents for the Neyman-Pearson tests}\label{sec:NP}

Let $\rho$ and $\sigma$ be quantum states such that
\begin{equation}\label{supports}
\supp\rho\subseteq\supp\sigma,
\end{equation}
and let $\rho_n,\sigma_n$, etc. be defined as in \eqref{eq:10}.
To exclude a trivial case, we assume that $\rho\ne\sigma$.
Let us define the quantum Neyman-Pearson tests by
\begin{align}
S_n(a):=\Set{\rho_n-e^{na}\sigma_n>0},
\label{eq:42}
\end{align}
where $a\in\R$ is a trade-off parameter. Our goal in this section is to determine the
asymptotics of the corresponding type I success probabilities $\Tr\rho_n S_{n,a}$ and the
type II error probabilities $\Tr\sigma_n S_{n,a}$.
Note that
\begin{equation}\label{zero NP}
S_n(a)=0\ds\iff\ds a\ge\dmax{\rho}{\sigma}:=\inf\{\gamma\,:\,\rho\le e^{\gamma}\sigma\}.
\end{equation}
Here $\dmax{\rho}{\sigma}$ is the \ki{max-relative entropy} \cite{Datta,RennerPhD}, and it was shown in \cite[Theorem 4]{Renyi_new} that 
\begin{align*}
D_{+\infty}\nw(\rho\|\sigma):=\lim_{\alpha\to+\infty}\rsrn{\rho}{\sigma}{\alpha}=\dmax{\rho}{\sigma}.
\end{align*}
Thus,
\begin{equation*}
\Tr\rho_n S_{n,a}=\Tr\sigma_n S_{n,a}=0,\ds\ds\ds a\ge \dmax{\rho}{\sigma},
\end{equation*}
and, with the convention $\log 0:=-\infty$,
\begin{equation*}
\lim_{n\to+\infty}\frac{1}{n}\log\Tr\rho_n S_{n,a}=\lim_{n\to+\infty}\frac{1}{n}\log\Tr\sigma_n S_{n,a}=-\infty,\ds\ds\ds a\ge \dmax{\rho}{\sigma}.
\end{equation*}
Hence, for the rest we can restrict our attention to $a<\dmax{\rho}{\sigma}$.

For every $s\in\bR$, let
\begin{align}
\psi(s):=F_{s+1}(\rho\|\sigma)
=\log\Tr\left(
\sigma^{\frac{-s}{2(s+1)}}\rho\sigma^{\frac{-s}{2(s+1)}}\right)^{s+1},
\label{eq:41}
\end{align}
and
\begin{align}
\phi(a):=\sup_{s\ge 0}\{as-\psi(s)\}
\label{eq:45}
\end{align}
be its Legendre-Fenchel transform on the interval $[0,+\infty)$.
\begin{lemma}\label{lemma:psi properties}
We have
\begin{align}
\psi(0)&=0,\label{psi(0)}\\
\psi'(0)&=\sr{\rho}{\sigma},\label{phi'(0)}\\
\lim_{s\to+\infty}\psi'(s)&=\dmax{\rho}{\sigma},\label{phi'(infty)}
\end{align}
and
\begin{equation}\label{phi values}
\phi(a)\begin{cases}
=0,&a\le\sr{\rho}{\sigma}\\
>0,&\sr{\rho}{\sigma}<a\le \dmax{\rho}{\sigma},\\
=+\infty,&\dmax{\rho}{\sigma}<a.
\end{cases}
\end{equation}
\end{lemma}
\begin{proof}
The identity in \eqref{psi(0)} is immediate from the definition of $\psi$.
$\psi(0)=0$ yields $\psi'(0)=\lim_{s\to 0}\frac{1}{s}\psi(s)=\lim_{\alpha\to 1}\rsr{\rho}{\sigma}{\alpha}=\sr{\rho}{\sigma}$, where the last identity is due to \cite[Theorem 4]{Renyi_new}.
Using again \cite[Theorem 4]{Renyi_new} and the L'Hospital rule,
$\lim_{s\to+\infty}\psi'(s)=\lim_{s\to+\infty}\frac{1}{s}\psi(s)=\lim_{\alpha\to+\infty}\rsr{\rho}{\sigma}{\alpha}=\dmax{\rho}{\sigma}$.
By Corollary \ref{cor:convexity}, $s\mapsto\psi(s)$ is convex, and hence \eqref{phi values} follows immediately from \eqref{psi(0)}--\eqref{phi'(infty)}.
\end{proof}

\begin{lemma}\label{lemma:NP upper bounds}
For any $a\in\R$ and $n\in\bN$, we have
\begin{align}
\frac{1}{n}\log\Tr\rho_nS_n(a)&\le -\phi(a),
\label{eq:46}\\
\frac{1}{n}\log\Tr\sigma_nS_n(a)
&\le
-\{a+\phi(a)\}.
\label{eq:65}
\end{align}
\end{lemma}
\begin{proof}
For any $a\in\R$ and $s\ge 0$, we have
\begin{align}
\Tr\rho_nS_n(a)
&=\left\{\Tr\rho_nS_n(a)\right\}^{s+1}
\left\{\Tr\rho_nS_n(a)\right\}^{-s}
\nn
&\le e^{-nas}
\left\{\Tr\rho_nS_n(a)\right\}^{s+1}\left\{\Tr\sigma_nS_n(a)\right\}^{-s}
\nn
&\le e^{-nas}\Bigl[
\left\{\Tr\rho_nS_n(a)\right\}^{s+1}\left\{\Tr\sigma_nS_n(a)\right\}^{-s}
\nn
&\quad+\left\{\Tr\rho_n(I_n-S_n(a))\right\}^{s+1}
\left\{\Tr\sigma_n(I_n-S_n(a))\right\}^{-s}
\Bigr]
\nn
&\le e^{-nas}
\Tr\left(
\sigma_n^{\frac{-s}{2(s+1)}}\rho_n\sigma_n^{\frac{-s}{2(s+1)}}
\right)^{s+1}
\nn
&= e^{-nas}e^{n\psi(s)},
\label{eq:44}
\end{align}
where in the first inequality we used \eqref{eq:39}, the second inequality is trivial, and the last inequality follows from Lemma~\ref{mono:measurement}.
Taking the logarithm and the infimum in $s$ yields
the inequality in \eqref{eq:46}.

Using \eqref{eq:39} and \eqref{eq:44}, we get
\begin{align}
\Tr\sigma_nS_n(a)\le e^{-na}\Tr\rho_nS_n(a)\le e^{-na(s+1)}e^{n\psi(s)},
\label{eq:64}
\end{align}
which yields \eqref{eq:65}.
\end{proof}

Note that the bounds in \eqref{eq:46} and \eqref{eq:65} are trivial for $a\ge \dmax{\rho}{\sigma}$, due to
\eqref{zero NP}. For $a\le\sr{\rho}{\sigma}$ we have $\phi(a)=0$ (cf.~\eqref{phi values}), and hence the upper bound in \eqref{eq:46} is trivial in this range. More detailed information
about the values of $\Tr\sigma_nS_n(a)$ in this range is given in the setting of the Hoeffding bound;
Corollary 4.5 in \cite{HMO2} states that
\begin{align*}
\lim_{n\to\infty}\frac{1}{n}\log\Tr\sigma_nS_n(a)=-\sup_{0\le t\le 1}\{at-\log\Tr\rho^t\sigma^{1-t}\}
\le -a=-\{\phi(a)+a\},\ds\ds\ds a<\sr{\rho}{\sigma}.
\end{align*}
Theorems \ref{thm:alpha} and \ref{thm:beta} below show that the inequalities in \eqref{eq:46} and \eqref{eq:65}  hold asymptotically as an equality in the non-trivial range $\sr{\rho}{\sigma}<a<\dmax{\rho}{\sigma}$.

\begin{theorem}
\label{thm:alpha}
For any $a\in\bz D(\rho\|\sigma),\dmax{\rho}{\sigma}\jz$, we have
\begin{align}
\lim_{n\to\infty}\frac{1}{n}\log\Tr\rho_nS_n(a)
=\lim_{n\to\infty}\frac{1}{n}\log\Tr(\rho_n-e^{na}\sigma_n)_+
= -\phi(a).
\label{eq:47}
\end{align}
\end{theorem}
\begin{proof}
For a fixed $m\in\bN$,
let $\widehat\rho_m:=\E_{\sigma_m}(\rho_m)$, and define
\begin{align}
\widehat S_{m,k}(a)
:=\Set{\widehat\rho_m^{\otimes k}-e^{kma}\sigma_m^{\otimes k}>0}.
\label{eq:48}
\end{align}
Write $n\in\bN$ in the form $n=km+r$, $k,r\in\bN,\,0\le r<m$.
For any $a,b\in\bR$, we have
\begin{align}
\Tr\rho_nS_n(a)
&=\Tr(\rho_n-e^{na}\sigma_n)S_n(a)+e^{na}\Tr\sigma_nS_n(a)
\nn
&\ge\Tr(\rho_n-e^{na}\sigma_n)_+
\nn
&\ge\Tr(\widehat\rho_m^{\otimes k}-e^{na}\sigma_m^{\otimes k})_+
\label{eq:49}
\\
&\ge\Tr(\widehat\rho_m^{\otimes k}-e^{na}\sigma_m^{\otimes k})
\widehat S_{m,k}(b)
\label{eq:50}
\\
&\ge\Tr\widehat\rho_m^{\otimes k}\widehat S_{m,k}(b)
-e^{na}e^{-kmb}\Tr\widehat\rho_m^{\otimes k}\widehat S_{m,k}(b)
\label{eq:51}
\\
&=\{1-e^{ra}e^{-km(b-a)}\}\Tr\widehat\rho_m^{\otimes k}\widehat S_{m,k}(b),
\label{eq:52}
\end{align}
where \eqref{eq:49} follows from Lemma~\ref{mono:positive}
(with the choice $\F:=\E_{\sigma_m}^{\otimes k}\otimes\Tr_{[km+1,r]})$,
\eqref{eq:50} follows from \eqref{eq:38},
and we used \eqref{eq:39} in \eqref{eq:51}.
Hence, by choosing $b>a$, we get
\begin{align}
-\phi(a)&\ge \limsup_{n\to+\infty}\frac{1}{n}\log\Tr\rho_nS_n(a)
\ge
\liminf_{n\to+\infty}\frac{1}{n}\log\Tr\rho_nS_n(a)\nonumber\\
&\ge
\liminf_{n\to\infty}\frac{1}{n}\log\Tr(\rho_n-e^{na}\sigma_n)_+
\ge
\frac{1}{m}\liminf_{k\to\infty}\frac{1}{k}\Tr\widehat\rho_m^{\otimes k}\widehat S_{m,k}(b),
\label{success prob lower bound}
\end{align}
where the first inequality is due to \eqref{eq:46}.

Note that $\widehat\rho_m$ and $\sigma_m$ are commuting density operators, and hence they can be represented as probability density
functions on some finite set $\X$, which is the interpretation we will be using in the following. Then
$Y:=\log\frac{\widehat\rho_m}{\sigma_m}$ is a random variable on $\X$, and its logarithmic moment generating function w.r.t.~
$\widehat \rho_m$ is
\begin{align}
m\psi_m(s):=\Psi_m(s)&:=
\log\Exp_{\widehat\rho_m}e^{ s\log\frac{\widehat\rho_m}{\sigma_m} }=
\log\Tr\widehat\rho_m e^{ s\log\frac{\widehat\rho_m}{\sigma_m} }
=\log\Tr\widehat\rho_m^{1+s}\sigma_m^{-s}.
\label{eq:53}
\end{align}
Note that $\log\frac{\widehat\rho_m^{\otimes k}}{\sigma_m^{\otimes k}}$ can naturally be identified with
$Y_1+\ldots+Y_k$, where $Y_i$ is the $i$th translate of $Y$ on $\times_{j=1}^{+\infty}\X$. Obviously, these translates form a
sequence of i.i.d.~random variables under the product law $\widehat\rho_m^{\otimes\infty}$, and hence, by
Cram\'er's theorem \cite[Theorem 2.1.24]{DZ}, we have
\begin{align*}
\liminf_{k\to\infty}\frac{1}{k}\log
\Tr\widehat\rho_m^{\otimes k}\widehat S_{m,k}(b)
=\liminf_{k\to\infty}\frac{1}{k}\log
\Tr\widehat\rho_m^{\otimes k}
\Set{\frac{1}{k}\log\frac{\widehat\rho_m^{\otimes k}}{\sigma_m^{\otimes k}}>mb}
\ge-\inf_{\kappa>mb}\sup_{s\in\R}\left\{\kappa s-\Psi_m(s) \right\}.
\end{align*}
Assume now that $\sr{\rho}{\sigma}<a<b<\dmax{\rho}{\sigma}$. Then we have
\begin{align*}
mb>mD(\rho\|\sigma)=\sr{\rho_m}{\sigma_m}\ge D(\widehat\rho_m\|\sigma_m)
=\Exp_{\widehat\rho_m}\log\frac{\widehat\rho_m}{\sigma_m}
=\Psi_m'(0),
\end{align*}
where the second inequality is due to the
monotonicity of the quantum relative entropy.
Since $\Psi_m$ is convex, it follows that
\begin{align*}
\inf_{\kappa>mb}\sup_{s\in\R}\left\{\kappa s-\Psi_m(s) \right\} &=
\sup_{s\in\R}\left\{mb s-\Psi_m(s) \right\}
=
\sup_{s\ge 0}\left\{mb s-\Psi_m(s) \right\}
=m\sup_{s\ge 0}\left\{ bs-\psi_m(s) \right\}.
\end{align*}
Let $\delta_m:=\frac{\log v_m}{m}$.
From \eqref{eq:27}, we obtain
\begin{align}
\psi(s)\le\psi_m(s)+(1+s)\delta_m,
\label{eq:58}
\end{align}
and hence,
\begin{align*}
\sup_{s\ge 0}\left\{ bs-\psi_m(s) \right\}
&\le\sup_{s\ge 0}\left\{ bs-\psi(s)+(1+s)\delta_m\right\}
\nn
&=\sup_{s\ge 0}\left\{ \left(b+\delta_m\right)s-\psi(s) \right\}+\delta_m
\nn
&\le\phi(b+\delta_m) +\delta_m.
\end{align*}
Putting it all together, we get
\begin{align}
\frac{1}{m}\liminf_{k\to\infty}\frac{1}{k}\log
\Tr\widehat\rho_m^{\otimes k}\widehat S_{m,k}(b)
\ge
-\left\{\phi(b+\delta_m) +\delta_m\right\}.
\label{eq:60}
\end{align}
Substituting it back to \eqref{success prob lower bound}, taking the limit $m\to+\infty$ and
using that $\lim_{m\to+\infty}\delta_m=0$, and that $\phi$ is continuous on
$\bz\sr{\rho}{\sigma},\dmax{\rho}{\sigma}\jz$, we obtain
the assertion.
\end{proof}

\begin{theorem}
\label{thm:beta}
For any $a\in\bz D(\rho\|\sigma),\dmax{\rho}{\sigma}\jz$, we have
\begin{align}
\lim_{n\to\infty}\frac{1}{n}\log\Tr\sigma_nS_n(a)
= -\{\phi(a)+a\}.
\label{eq:63}
\end{align}
\end{theorem}
\begin{proof}
By \eqref{eq:38}, we have
\begin{align}
\Tr(\rho_n-e^{nb}\sigma_n)_+
\ge \Tr(\rho_n-e^{nb}\sigma_n)S_n(a)
\label{eq:66}
\end{align}
for any $b\in\R$, and hence,
\begin{align}
\Tr(\rho_n-e^{nb}\sigma_n)_+ +e^{nb}\Tr\sigma_nS_n(a)
\ge \Tr\rho_nS_n(a).
\label{eq:67}
\end{align}
Assume now that $\sr{\rho}{\sigma}<a<b<\dmax{\rho}{\sigma}$.
Applying Theorem~\ref{thm:alpha} to
\eqref{eq:67}, we get
\begin{align*}
-\phi(a)&=
\liminf_{n\to\infty}\frac{1}{n}\log\Tr\rho_nS_n(a)
\le
\max\left\{-\phi(b),b+\liminf_{n\to\infty}\frac{1}{n}\log\Tr\sigma_nS_n(a)\right\}.
\end{align*}
Note that $\sr{\rho}{\sigma}<a<b<\dmax{\rho}{\sigma}$ implies $\phi(a)<\phi(b)$, and hence we have
\begin{align}
-\phi(a)\le b+\liminf_{n\to\infty}\frac{1}{n}\log\Tr\sigma_nS_n(a).
\label{eq:68}
\end{align}
Taking $b\searrow a$, we obtain
\begin{align}
-\{\phi(a)+a\}\le\liminf_{n\to\infty}\frac{1}{n}\log\Tr\sigma_nS_n(a).
\label{eq:70}
\end{align}
Now combining \eqref{eq:65} and \eqref{eq:70} yields the assertion.
\end{proof}
\medskip

Theorems \ref{thm:alpha} and \ref{thm:beta} yield the following refinement of Corollary \ref{cor:Renyi attainability}.
Note that $\phi(a)$ can also be written as
$\phi(a)=\sup_{\alpha>1}\{a(\alpha-1)-F_{\alpha}(\rho\|\sigma)\}$,
where $F_{\alpha}(\rho\|\sigma)$ is defined in \eqref{eq:5}. For simplicity, we will use the notation
$F(\alpha):=F_{\alpha}(\rho\|\sigma)$. By Corollary \ref{cor:convexity}, $\alpha\mapsto F(\alpha)$ is convex on $(1,+\infty)$,
and Lemma \ref{lemma:psi properties} yields that for every $\alpha\in(1,+\infty)$ there exists an $a_{\alpha}\in(D(\rho\|\sigma),\dmax{\rho}{\sigma})$ such that
\begin{equation}\label{a alpha}
\phi(a_{\alpha})=a_{\alpha}(\alpha-1)-F(\alpha).
\end{equation}

\begin{corollary}\label{cor:achievability by NP}
For every $\alpha>1$, let $a_{\alpha}$ be as above, and let
$p_{n,\alpha}:=\{\Tr\rho_n S_{n}(a_{\alpha}),\Tr\rho_n(I_n- S_{n}(a_{\alpha}))\}$,
$q_{n,\alpha}:=\{\Tr\sigma_n S_{n}(a_{\alpha}),\Tr\sigma_n(I_n- S_{n}(a_{\alpha}))\}$
be the post-measurement states corresponding to the Neyman-Pearson test $S_n(a_{\alpha})$. Then
\begin{equation*}
\lim_{n\to+\infty}\frac{1}{n}\rsr{p_{n,\alpha}}{q_{n,\alpha}}{\alpha}=\rsrn{\rho}{\sigma}{\alpha}.
\end{equation*}
\end{corollary}
\begin{proof}
Omitting a standard $\ep-\delta$ argument, we can write Theorems \ref{thm:alpha} and \ref{thm:beta} as
$\Tr\rho_n S_{n}(a_{\alpha})\sim e^{-n\phi(a_{\alpha})}$ and
$\Tr\sigma_n S_{n}(a_{\alpha})\sim e^{-n(\phi(a_{\alpha})+a_{\alpha})}$, which then yields
\begin{equation*}
\bz\Tr\rho_n S_{n}(a_{\alpha})\jz^{\alpha}\bz\Tr\sigma_n S_{n}(a_{\alpha})\jz^{1-\alpha}
\sim
\exp\bz-n\left[\alpha\phi(a_{\alpha})+(1-\alpha)(\phi(a_{\alpha})+a_{\alpha})\right]\jz
=
\exp(n F(\alpha)),
\end{equation*}
where the last identity is due to \eqref{a alpha}. Note that $F(\alpha)>0$ for $\alpha>1$, and
$\lim_{n\to+\infty}\Tr\rho_n(I_n- S_{n}(a_{\alpha}))=
\lim_{n\to+\infty}\Tr\sigma_n(I_n- S_{n}(a_{\alpha}))=1$. Hence,
$Q_{\alpha}\nw(p_{n,\alpha}\|q_{n,\alpha})\sim \exp(nF(\alpha))$,
from which the assertion follows.
\end{proof}

\subsection{The strong converse exponent}

Consider the hypothesis testing problem from Section \ref{sec:hypotesting}. Our aim here is
to prove the identity \eqref{sc exponent equals converse Hoeffding}, i.e., that the
strong converse exponent $B_e^*(r)$, defined in \eqref{eq:72}, is equal to the converse Hoeffding bound $H_r^*(\rho\|\sigma)$ defined in
\eqref{converse Hoeffding}.
We will assume that $\rho\ne\sigma$ to avoid a trivial case, and that
$\supp\rho\subseteq\supp\sigma$ so that we actually have a strong converse (cf.~Remark
\ref{rem:supports}).

We start with the following lemma, which is a direct analogue of Nagaoka's proof of the strong converse to the quantum Stein's lemma \cite{Nagaoka2},
except that we use the new R\'enyi divergences instead of the old ones.

\begin{lemma}\label{lemma:converse rate lower bound}
For any $r\ge 0$, we have
\begin{align}
B_e^*(r)\ge 
H_r^*(\rho\|\sigma).
\label{sc rate lower bound}
\end{align}
\end{lemma}
\begin{proof}
Let $T_n\in\L(\hil_n)$ be a test and let $p_n:=\bz\Tr\rho_n T_n,\Tr\rho_n (I-T_n)\jz$ and
$q_n:=\bz\Tr\sigma_n T_n,\Tr\sigma_n (I-T_n)\jz$ be the post-measurement states.
By the monotonicity of the R\'enyi relative entropies under measurements (Lemma \ref{mono:measurement}), we have, for any
$\alpha>1$,
\begin{align*}
\rsrn{\rho_n}{\sigma_n}{\alpha}
\ge
\rsrn{p_n}{q_n}{\alpha}
&\ge
\frac{1}{\alpha-1}\log\left[(\Tr\rho_n T_n)^{\alpha}(\Tr\sigma_n T_n)^{1-\alpha}\right]\\
&=
\frac{\alpha}{\alpha-1}\log(1-\alpha_n(T_n))-\log\beta_n(T_n),
\end{align*}
or equivalently,
\begin{align}\label{Nagaoka bound}
\frac{1}{n}\log(1-\alpha_n(T_n))
\le
\frac{\alpha-1}{\alpha}\left[\rsrn{\rho}{\sigma}{\alpha}+\frac{1}{n}\log\beta_n(T_n)\right].
\end{align}
If $\limsup_{n\to\infty}\frac{1}{n}\log\Tr\sigma_nT_n \le -r$
then
\begin{align*}
\limsup_{n\to\infty}\frac{1}{n}\log(1-\alpha_n(T_n))
\le
\frac{\alpha-1}{\alpha}\left[\rsrn{\rho}{\sigma}{\alpha}-r\right],\ds\ds\ds \alpha>1.
\end{align*}
Taking the infimum in $\alpha>1$, the statement follows.
\end{proof}

\begin{remark}
Using that the old R\'enyi relative entropies are also monotonic under measurements
\cite{H:text}, exactly the same argument as above yields that
\begin{align}\label{HN bound}
B_e^*(r)\ge
\sup_{1<\alpha}\frac{\alpha-1}{\alpha}\left[r-\rsro{\rho}{\sigma}{\alpha}
\right].
\end{align}
This was already pointed out in \cite{ON} with a restricted optimization over
$\alpha\in(1,2]$, and later extended by Hayashi to the above form \cite{H:text}.
\end{remark}
\medskip

Our goal in the rest of the section is to show that \eqref{sc rate lower bound} holds as an equality.
To start with, we give some alternative expressions for $H_r^*(\rho\|\sigma)$.
Let
\begin{equation}\label{amax}
a_{\max}:=\dmax{\rho}{\sigma},\ds\ds\text{and}\ds\ds
r_{\max}:=\phi(a_{\max})+a_{\max}.
\end{equation}
Note that
\begin{equation}\label{converse Hoeffding 1}
H_r^*(\rho\|\sigma)=\sup_{s\ge 0}\frac{rs-\psi(s)}{s+1}=\sup_{0\le u<1}\{ur-\tilde\psi(u)\},
\end{equation}
where
\begin{equation*}
\tilde\psi(u):=(1-u)\psi\bz \frac{u}{1-u}\jz,\ds\ds\ds u\in[0,1).
\end{equation*}
It is easy to see that $\tilde\psi'(u)=-\psi(s)+(1+s)\psi'(s)$ with the notational convention
$u=s/(s+1)$, and hence
\begin{equation}\label{psi at 0}
\tilde\psi(0)=\psi(0)=0,\ds\ds\ds
\tilde\psi'(0)=\psi'(0)=\sr{\rho}{\sigma},
\end{equation}
and
\begin{align*}
\lim_{u\nearrow 1}\tilde\psi'(u)&=\lim_{s\to+\infty}\bz s\psi'(s)-\psi(s)\jz+
\lim_{s\to+\infty}\psi'(s)=
\lim_{s\to+\infty}\phi\bz\psi'(s)\jz+\dmax{\rho}{\sigma}=
\phi(a_{\max})+a_{\max}\\
&=r_{\max}.
\end{align*}
It is also easy to see, by computing the second derivative, that $\tilde\psi$ is convex for
commuting $\rho$ and $\sigma$; convexity in the general case then follows the same way as in
Corollary \ref{cor:convexity}. Convexity and \eqref{psi at 0} yield
\begin{equation}\label{converse Hoeffding nullity}
H_r^*(\rho\|\sigma)=0,\ds\ds\ds r\le \sr{\rho}{\sigma}.
\end{equation}

\begin{lemma}\label{lemma:converse Hoeffding repr}
For any $r\ge 0$, we have
\begin{align}
H_r^*(\rho\|\sigma)
=\begin{cases}
r-a_r=\phi(a_r), & r< \phi(a_{\max})+a_{\max}, \\
r - \dmax{\rho}{\sigma}, & r\ge \phi(a_{\max})+a_{\max},
\end{cases}
\label{eq:73-3}
\end{align}
where $a_{\max}$ and $r_{\max}$ are defined in \eqref{amax}, and
$a_r$ is the unique solution of $r-a_r=\phi(a_r)$.
\end{lemma}
\begin{proof}
First, we consider the case $0\le r<r_{\max}$.
Note that $a\mapsto \phi(a)+a$ is strictly increasing and continuous on $(-\infty,a_{\max})$, and hence
for every $r<r_{\max}$ there exists a unique $a_r$ such that $r=\phi(a_r)+a_r$
By definition,
\begin{equation*}
\phi(a_r)\ge a_r s-\psi(s)=s(r-\phi(a_r))-\psi(s),\ds\ds\ds s\ge 0,
\end{equation*}
and equality holds in the above inequality for some $s_r\in[0,+\infty)$.
Rearranging, we get
\begin{equation*}
\phi(a_r)\ge\frac{sr-\psi(s)}{1+s},\ds\ds\ds s\ge 0,
\end{equation*}
with equality for $s_r$, and hence
\begin{equation*}
\phi(a_r)=\max_{s\ge 0}\frac{sr-\psi(s)}{1+s}.
\end{equation*}
Taking into account \eqref{converse Hoeffding 1}, this proves the assertion.

Next, assume that $r\ge r_{\max}$.
Note that
\begin{align}
\lim_{s\to +\infty}\frac{rs-\psi(s)}{s+1} =
r-\lim_{s\to +\infty}\frac{\psi(s)}{s+1}
 =r-\dmax{\rho}{\sigma},
\end{align}
due to \cite[Theorem 4]{Renyi_new}.
Hence it is enough to show that
\begin{align}
\frac{rs-\psi(s)}{s+1}\le r-\dmax{\rho}{\sigma}
\end{align}
for every $s\ge 0$. 
Note that $r\ge r_{\max}=\phi(a_{\max})+a_{\max}$ implies
\begin{align}
r-a_{\max}\ge\phi(a_{\max})\ge a_{\max}s-\psi(s)
\end{align}
for every $s\ge 0$, from which we obtain
\begin{align}
\frac{r+\psi(s)}{s+1}\ge a_{\max}.
\end{align}
Thus we have
\begin{align}
r-a_{\max}
\ge r-\frac{r+\psi(s)}{s+1}
=\frac{rs-\psi(s)}{s+1},
\end{align}
and hence $H_r^*(\rho\|\sigma)=r-\dmax{\rho}{\sigma}$, as required.
\end{proof}

Now we are ready to prove the identity \eqref{sc exponent equals converse Hoeffding}
for the strong converse exponent.

\begin{theorem}
\label{thm:exponent}
For any $r\ge 0$, we have
\begin{align}
B_e^*(r)
&=H_r^*(\rho\|\sigma).
\end{align}
\end{theorem}
\begin{proof}
Since we have already shown $B_e^*(r)\ge H_r^*(\rho\|\sigma)$
in Lemma \ref{lemma:converse rate lower bound}, we only have to show the converse inequality
$B_e^*(r)\le H_r^*(\rho\|\sigma)$. Due to the definition \eqref{sc rate def2} of $B^*_e(r)$ as an infimum of rates, this is
equivalent to showing
that for any rate $R>H_r^*(\rho\|\sigma)$
there exists a sequence of tests $\{T_n\}_{n=1}^{\infty}$ satisfying
\begin{align}\label{semioptimal tests}
\limsup_{n\to\infty}\frac{1}{n}\log\Tr\sigma_nT_n \le -r
\ds\ds\ds\ds\ds\text{and}\ds\ds\ds\ds\ds
\liminf_{n\to\infty}\frac{1}{n}\log\Tr\rho_n T_n\ge -R.
\end{align}
We prove the claim by considering three different regions of $r$.
\begin{enumerate}
\item[(i)]
In the case $D(\rho\|\sigma)< r<r_{\max}$,
there exists a unique $a_r\in(\sr{\rho}{\sigma},\dmax{\rho}{\sigma})$ satisfying $r-a_r=\phi(a_r)$,
and Theorems \ref{thm:alpha} and \ref{thm:beta} yield
\begin{align*}
\lim_{n\to\infty}\frac{1}{n}\log\Tr\sigma_nS_n(a_{r}) &=-(\phi(a_r)+a_r)= -r,\\
\lim_{n\to\infty}\frac{1}{n}\log\Tr\rho_nS_n(a_{r}) &= -\phi(a_{r})=
- H_r^*(\rho\|\sigma),
\end{align*}
where the last identity is due to Lemma \ref{lemma:converse Hoeffding repr}.

\item[(ii)]
In the case $0\le r\le D(\rho\|\sigma)$, we have $H_r^*(\rho\|\sigma)=0$, according to
\eqref{converse Hoeffding nullity}.
For any $R>0$, we can find an $a\in(\sr{\rho}{\sigma},\dmax{\rho}{\sigma})$
such that $0<\phi(a)<R$. Note that $\phi(a)+a>D(\rho\|\sigma)\ge r$, and
Theorems \ref{thm:alpha} and \ref{thm:beta} yield
\begin{align*}
\lim_{n\to\infty}\frac{1}{n}\log\Tr\sigma_nS_n(a) &=-(\phi(a)+a)< -r,\\
\lim_{n\to\infty}\frac{1}{n}\log\Tr\rho_nS_n(a) &= -\phi(a)>-R.
\end{align*}

\item[(iii)] In the case $r\ge r_{\max}$,
we use a modification of the Neyman-Pearson tests,
following the method of the proof of Theorem 4 in \cite{NH2007}.
For every $a,r\in\bR$, let
\begin{align*}
T_{n}(r,a):=e^{-n\{r-a-\phi(a)\}}S_n(a).
\end{align*}
Note that for $r\ge r_{\max}$ we have
$H_r^*(\rho\|\sigma)=r-\dmax{\rho}{\sigma}$ due to Lemma \ref{lemma:converse Hoeffding repr}. Assume now that
$a\in\bz\sr{\rho}{\sigma},\dmax{\rho}{\sigma}\jz$. Then
$r>\phi(a)+a$, and hence $0\le T_{n}(r,a)\le I$, i.e., $T_{n}(r,a)$ is a test, and
\begin{align*}
\lim_{n\to\infty}\frac{1}{n}\log\Tr\sigma_n T_n(r,a)
&=-r+a+\phi(a)-(a+\phi(a))= -r,\\
\lim_{n\to\infty}\frac{1}{n}\log\Tr\rho_n T_n(r,a)
&=-r+a+\phi(a)-\phi(a)=-(r-a),
\end{align*}
by Theorems \ref{thm:alpha} and \ref{thm:beta}. Now for a given $R>H_r^*(\rho\|\sigma)=r-\dmax{\rho}{\sigma}$, we can find
an $a\in\bz\sr{\rho}{\sigma},\dmax{\rho}{\sigma}\jz$ such that $r-\dmax{\rho}{\sigma}<r-a<R$, and the assertion follows.
\end{enumerate}
\end{proof}

\begin{remark}\label{rem:diagonal}
It is easy to see, by applying a standard diagonal argument, that there exists a sequence of tests $\{T_n\}_{n\in\bN}$ such
that \eqref{semioptimal tests} holds with $H_r^*(\rho\|\sigma)$ in place of $R$, and the proof of
Theorem \ref{thm:exponent} yields that for this sequence, we actually have
\begin{align*}
\limsup_{n\to\infty}\frac{1}{n}\log\Tr\sigma_nT_n \le -r
\ds\ds\ds\ds\ds\text{and}\ds\ds\ds\ds\ds
\liminf_{n\to\infty}\frac{1}{n}\log\Tr\rho_n T_n=-H_r^*(\rho\|\sigma).
\end{align*}
Moreover, it is also possible to have
$\limsup_{n\to\infty}\frac{1}{n}\log\Tr\sigma_nT_n=-r$ above;
this is obvious in cases (i) and (iii) in the proof of Theorem \ref{thm:exponent}, and in
case (ii) this follows from the Hoeffding bound theorem \cite{Hayashi,Nagaoka}.
\end{remark}

\begin{remark}
The direct region $(0\le r<D(\rho\|\sigma))$
and the strong converse region $(r>D(\rho\|\sigma))$
in quantum hypothesis testing are considered to be dual,
and the theory of both regions can be developed logically independently of the other,
which is the approach that we followed here.

Following a different approach, one could prove
$B_e^*(r)\le H_r^*(\rho\|\sigma)$ in the case $0\le r<D(\rho\|\sigma)$ (case (ii) of the above proof)
based on Stein's lemma rather than our argument.
Indeed, applying \eqref{eq:39} with $a=r$, we  have $\Tr\sigma_nS_n(a) \le e^{-nr}$,
and at the same time,
the direct part of the quantum Stein's lemma \cite{HP}  yields
$\lim_{n\to\infty}\Tr\rho_nS_n(a)=1$.
Thus,
\begin{align*}
\limsup_{n\to\infty}\frac{1}{n}\log\Tr\sigma_nS_n(a) \le -r
\ds\ds\ds\ds\ds\text{and}\ds\ds\ds\ds\ds
\liminf_{n\to\infty}\frac{1}{n}\log\Tr\rho_nS_n(a)=0=H_r^*(\rho\|\sigma).
\end{align*}
\end{remark}

\begin{remark}
By Theorem \ref{thm:exponent} and \eqref{converse Hoeffding 1}, we have
\begin{equation*}
B_e^*(r)=H_r^*(\rho\|\sigma)=\sup_{0\le u<1}\{ru-\tilde\psi(u)\},
\end{equation*}
where $\tilde\psi(u)$
is a continuous convex function on $[0,1)$. Hence, $B_e^*(r)$ is the Legendre-Fenchel transform (polar function) of
$\tilde\psi$, and the bipolar theorem says that
\begin{equation}\label{inverse Legendre}
\sup_{r\ge 0}\{ur-B_e^*(r)\}=\tilde\psi(u)=\frac{\alpha-1}{\alpha}\rsrn{\rho}{\sigma}{\alpha},\ds\ds\ds \alpha>1,
\end{equation}
where in the last formula we set $\alpha:=1/(1-u)$ and used the definition \eqref{eq:41} of $\psi$. That is, the new R\'enyi relative entropies can be expressed essentially as the Legendre-Fenchel transform of the operational quantities $B_e^*(r),\,r\ge 0$.
A more direct operational interpretation is provided in the next section.
\end{remark}

\begin{remark}\label{Remark-Hayashi-2}
A possible proof for the following representation of the strong converse exponent:
\begin{align}
B_e^*(r)=\max_{s\ge 0}\frac{rs-\lim_{m\to\infty}\psi_m(s)}{s+1},
\label{eq:87}
\end{align}
where $\psi_m$ is defined in \eqref{eq:53},
has been outlined in Hayashi's book \cite{H:text}, although
it seems to have not been fully worked out. Apart from identifying the limit
$\lim_{m\to\infty}\psi_m(s)$ as $s\rsrn{\rho}{\sigma}{1+s}$, our approach here differs
from Hayashi's proposal also in that we prove the achievability part by computing explicitly the asymptotic error rates
of the Neyman-Pearson tests, providing yet another operational interpretation for the
new R\'enyi divergences.
\end{remark}
\medskip

We note that Theorem \ref{thm:exponent} yields an operational proof of the 
Lieb-Thirring inequality.
Indeed, combining \eqref{HN bound} with \eqref{inverse Legendre}, we get that
\begin{equation*}
\rsro{\rho}{\sigma}{\alpha}
\ge\rsrn{\rho}{\sigma}{\alpha},\ds\ds\ds\alpha>1,
\end{equation*}
or equivalently,
\begin{equation*}
\Tr\rho^{\alpha}\sigma^{1-\alpha}\ge
\Tr\bz\rho^{\half}\sigma^{\frac{1-\alpha}{\alpha}}\rho^{\half}\jz^{\alpha},\ds\ds\ds\alpha>1.
\end{equation*}
Introducing $A:=\rho^{\half}$ and $B:=\sigma^{\frac{1-\alpha}{\alpha}}$, the above can be rewritten as
\begin{equation}\label{LT inequality}
\Tr A^{\alpha}B^{\alpha}A^{\alpha}\ge
\Tr\bz ABA\jz^{\alpha},\ds\ds\ds\alpha>1.
\end{equation}
Since we were interested in hypothesis testing, we only derived Theorem \ref{thm:exponent}
for density operators; however, it is easy to see that it also holds, with obvious modifications, for arbitrary positive semidefinite operators. Hence we arrive at the following:
\begin{corollary}
[Lieb-Thirring inequality]
For any positive semidefinite operators $A$ and $B$,
\eqref{LT inequality} holds.
\end{corollary}
\medskip

To close the section, we give one more representation of $H_r^*(\rho\|\sigma)$. This is closely related to the information spectrum approach \cite{NH2007}, and although we didn't need it in our proof for the strong converse exponent, an alternative proof could be given based on this representation.

\begin{lemma}
For any $r\ge 0$, we have
\begin{align}
H_r^*(\rho\|\sigma)
&=\inf_{a\in\R}\max\{\phi(a),r-a\}\label{eq:73-2}
\\
&=
\inf\Set{ \max\{\phi(a),r-a\} | \sr{\rho}{\sigma}<a<\dmax{\rho}{\sigma} }.
\label{eq:73-5}
\end{align}
\end{lemma}
\begin{proof}
Let $a_{\max}$ and $r_{\max}$ as in \eqref{amax}.
First, we consider the case $0\le r<r_{\max}$.
Let $a_r$ be the unique solution of $r=\phi(a_r)+a_r$, as in the proof of Lemma \ref{lemma:converse Hoeffding repr}.
Then
\begin{equation*}
\max\{\phi(a_r),r-a_r\}=\phi(a_r)=r-a_r.
\end{equation*}
Now if $a<a_r$ then $r-a>r-a_r$ and $\phi(a)\le\phi(a_r)$, which implies
$\max\{\phi(a),r-a\}=r-a>r-a_r$.
On the other hand, if $a>a_r$ then $r-a<r-a_r$, while $\phi(a)\ge\phi(a_r)$, and hence
$\max\{\phi(a),r-a\}=\phi(a)\ge\phi(a_r)$. Thus
\begin{align}
R(r):=\inf_{a\in\bR}\max\{\phi(a),r-a\}=
\max\{\phi(a_r),r-a_r\}=
\phi(a_r)=r-a_r,
\label{eq:78}
\end{align}
and \eqref{eq:73-2} follows by taking into account \eqref{eq:73-3}.

Note that when $\sr{\rho}{\sigma}<r<r_{\max}$ then $\sr{\rho}{\sigma}<a_r<\dmax{\rho}{\sigma}$, and
\eqref{eq:73-5} is immediate from \eqref{eq:78}.
In the case $0\le r\le D(\rho\|\sigma)$,
we have $r=a_r$ and $R(r)=\phi(a_r)=r-a_r=0$.
On the other hand, for every $\sr{\rho}{\sigma}<a<\dmax{\rho}{\sigma}$ we have
$\phi(a)> 0>r-a$, and thus
\begin{align*}
\inf\Set{ \max\{\phi(a),r-a\} | \sr{\rho}{\sigma}<a<\dmax{\rho}{\sigma} }
&=\inf\Set{ \phi(a) | \sr{\rho}{\sigma}<a<\dmax{\rho}{\sigma} }\\
&=0=R(r),
\end{align*}
proving \eqref{eq:73-5}.

Next, assume that $r\ge r_{\max}$. Then $r\ge\phi(a)+a$, or equivalently,
$r-a\ge\phi(a)$ for every $a\le a_{\max}$, and hence
$\max\{\phi(a),r-a\}=r-a$ for $a\le a_{\max}$, while for $a> a_{\max}$ we have
$\max\{\phi(a),r-a\}=\phi(a)=+\infty$. Hence,
\begin{align*}
R(r)=\inf_{a\in\bR}\max\{\phi(a),r-a\}&=
\inf\Set{ \max\{\phi(a),r-a\} | \sr{\rho}{\sigma}<a<\dmax{\rho}{\sigma} }\\
&=\inf_{a\le a_{\max}}\{r-a\}=r-a_{\max}=r-\dmax{\rho}{\sigma}.
\end{align*}
Taking into account \eqref{eq:73-3}, we get \eqref{eq:73-2} and \eqref{eq:73-5}.
\end{proof}

\subsection{Representation as cutoff rates}

In the setting of Section \ref{sec:hypotesting},
let
\begin{equation*}
\alpha_{n,r}:=\alpha_{e^{-nr}}(\rho^{\otimes n}\|\sigma^{\otimes n}):=\min\{\Tr\rho_n (I-T)\,:\,0\le T\le I,\,\Tr\sigma_n T\le e^{-nr}\}.
\end{equation*}
Following \cite{Csiszar}, we define the \ki{generalized $\kappa$-cutoff rate}
$C_{\kappa}(\rho\|\sigma)$ for any
$\kappa>0$ as the smallest $r_0$ such that
\begin{equation}\label{cutoff def}
\limsup_{n\to\infty}\frac{1}{n}\log(1-\alpha_{n,r})\le-\kappa(r-r_0),\ds\ds\ds r>0.
\end{equation}
As before, we assume that $\supp\rho\subseteq\supp\sigma$ and $\rho\ne\sigma$. 

\begin{lemma}\label{lemma:alphanr}
For every $r>0$, 
\begin{align*}
\lim_{n\to+\infty}\frac{1}{n}\log(1-\alpha_{n,r})=-H_r^*(\rho\|\sigma).
\end{align*}
\end{lemma}
\begin{proof}
Consider the inequality \eqref{Nagaoka bound}.
Taking the supremum over all test $T_n$ such that $\Tr\sigma_n T_n\le e^{-nr}$, we get 
\begin{align*}
\frac{1}{n}\log(1-\alpha_{n,r})
\le
\frac{\alpha-1}{\alpha}\left[\rsrn{\rho}{\sigma}{\alpha}-r\right].
\end{align*}
Taking now the limsup in $n$ and the infimum in $\alpha$, we obtain
\begin{align}\label{cutoff limsup}
\limsup_{n\to+\infty}\frac{1}{n}\log(1-\alpha_{n,r})
\le
-H_r^*(\rho\|\sigma).
\end{align}

According to Remark \ref{rem:diagonal}, for every $r'>0$,
there exists a sequence of tests $T_{n,r'}$,$\,n\geq1$, such that
\begin{align}\label{alphanr1}
\limsup_{n\to+\infty}\frac{1}{n}\log\Tr\sigma_nT_{n,r'}
\le
-r'\ds\ds\ds\text{ and }\ds\ds\ds
\liminf_{n\to+\infty}\frac{1}{n}\log\Tr\rho_n T_{n,r'} 
\ge
-H_{r'}(\rho\|\sigma).
\end{align}
Hence, for any $r'>r$, there exists an $N_{r'}$ such that for all $n>N_{r'}$, 
$\Tr\sigma_nT_{n,r'}\le e^{-nr}$, and thus $\Tr\rho_n T_{n,r'}\le 1-\alpha_{n,r}$. By the second inequality in 
\eqref{alphanr1},
\begin{equation}
\label{sc proof1}
\liminf_{n\to+\infty}\frac{1}{n}\log(1-\alpha_{n,r}) 
\ge
\liminf_{n\to+\infty}\frac{1}{n}\log\Tr\rho_n T_{n,r'} 
\ge
-H_{r'}(\rho\|\sigma).
\end{equation}
From the definition \eqref{converse Hoeffding} of the converse Hoeffding divergence,
it is clear that $r\mapsto H_{r}^*(\rho\|\sigma)$ is a monotone increasing convex function on $(0,+\infty)$.
Moreover, Lemma \ref{lemma:converse Hoeffding repr} implies that 
$H_{r}^*(\rho\|\sigma)$ is finite for every $r>0$. Thus,
$r\mapsto H_{r}^*(\rho\|\sigma)$ is
continuous on $(0,+\infty)$, and \eqref{sc proof1} yields
\begin{align}\label{cutoff liminf}
\liminf_{n\to+\infty}\frac{1}{n}\log(1-\alpha_{n,r}) 
\ge
\sup_{r'>r}-H_{r'}(\rho\|\sigma)=-H_{r}(\rho\|\sigma).
\end{align}
Finally, \eqref{cutoff limsup} and \eqref{cutoff liminf} yield the assertion.
\end{proof}

\begin{theorem}\label{thm:cutoff}
For every $\kappa\in(0,1)$,
\begin{equation*}
C_{\kappa}(\rho\|\sigma)=\rsrn{\rho}{\sigma}{\frac{1}{1-\kappa}}.
\end{equation*}
\end{theorem}
\begin{proof}
By Lemma \ref{lemma:alphanr} and \eqref{converse Hoeffding 1}, we have
\begin{equation*}
\lim_{n\to\infty}\frac{1}{n}\log(1-\alpha_{n,r})
=
-H_r^*(\rho\|\sigma)=-\sup_{0\le u<1}\{ru-\tilde\psi(u)\}.
\end{equation*}
By definition, we have
\begin{equation*}
H_r^*(\rho\|\sigma)\ge r\kappa-\tilde\psi(\kappa)=\kappa\bz r-\frac{1}{\kappa}\tilde\psi(\kappa) \jz,
\end{equation*}
and the above inequality holds with equality for $r_{\kappa}:=\tilde\psi'(\kappa)$, and hence
\begin{align*}
\frac{1}{\kappa}\tilde\psi(\kappa)&=
\frac{1}{\kappa}(1-\kappa)\psi\bz\frac{\kappa}{1-\kappa}\jz
=
\rsrn{\rho}{\sigma}{\frac{1}{1-\kappa}}
\end{align*}
is the smallest $r_0$ for which \eqref{cutoff def} holds.
\end{proof}

The above Theorem immediately yields the following operational interpretation of the new R\'enyi relative entropies:
\begin{corollary}\label{cor:cutoff}
For every $\alpha>1$,
\begin{equation*}
\rsrn{\rho}{\sigma}{\alpha}=C_{\frac{\alpha-1}{\alpha}}(\rho\|\sigma).
\end{equation*}
\end{corollary}
\bigskip

The above operational interpretation yields as an immediate consequence 
an alternative proof for 
the monotonicity of the new R\'enyi divergences, Corollary \ref{Renyi monotonicity} and Remark \ref{rem:EPPMON}:
\begin{corollary}\label{cor:EPPMON}
Let $\rho,\sigma\in\B(\hil)_+$ and $\map:\,\B(\hil)\to\B(\kil)$ be a trace-preserving linear map such that 
$\map^{\otimes n}$ is positive for every $n\in\bN$. Then
\begin{align*}
D_{\alpha}\nw(\map(\rho)\|\map(\sigma))\le D_{\alpha}\nw(\rho\|\sigma),\ds\ds\ds \alpha>1.
\end{align*}
In particular, $D_{\alpha}\nw$ is monotone non-increasing under CPTP maps for every $\alpha>1$.
\end{corollary}
\begin{proof}
By assumption, the Hilbert-Schmidt dual $(\map^{\otimes n})^*$ is positive and unital for every $n\in\bN$, and hence
\begin{align*}
\alpha_{e^{-nr}}(\map(\rho)^{\otimes n}\|\map(\sigma)^{\otimes n})&=
\min\{\Tr\map^{\otimes n}(\rho^{\otimes n}) (I-T)\,:\,0\le T\le I,\,\Tr\map^{\otimes n}(\sigma^{\otimes n}) T\le e^{-nr}\}
\\
&=
\min\{\Tr\rho^{\otimes n} (I-(\map^{\otimes n})^*(T))\,:\,0\le T\le I,\,\Tr\sigma^{\otimes n} (\map^{\otimes n})^*(T)\le e^{-nr}\}\\
&\ge
\min\{\Tr\rho^{\otimes n} (I-T)\,:\,0\le T\le I,\,\Tr\sigma^{\otimes n} T\le e^{-nr}\}\\
&=
\alpha_{e^{-nr}}(\rho^{\otimes n}\|\sigma^{\otimes n}).
\end{align*}
Thus for every $\kappa\in(0,1)$, and every $r>0$,
\begin{align*}
\limsup_{n\to+\infty}\frac{1}{n}\log (1-\alpha_{e^{-nr}}(\map(\rho)^{\otimes n}\|\map(\sigma)^{\otimes n}))
\le
\limsup_{n\to+\infty}\frac{1}{n}\log (1-\alpha_{e^{-nr}}(\rho^{\otimes n}\|\sigma^{\otimes n}))
\le
-\kappa r+\kappa \rsrn{\rho}{\sigma}{\frac{1}{1-\kappa}},
\end{align*}
where in the last inequality we used Theorem \ref{thm:cutoff}. By the definition of the $\kappa$-cutoff rate
and Theorem \ref{thm:cutoff}, we get
\begin{align*}
\rsrn{\map(\rho)}{\map(\sigma)}{\frac{1}{1-\kappa}}=C_{\kappa}(\map(\rho)\|\map(\sigma))
\le
\rsrn{\rho}{\sigma}{\frac{1}{1-\kappa}},
\end{align*}
proving the assertion.
\end{proof}

\section{Conclusion}

In this paper we have determined the exact strong converse exponent for binary quantum hypothesis testing, and showed that it can be expressed in terms of the recently introduced version of quantum
R\'enyi $\alpha$-relative entropies $D_{\alpha}\nw$ \cite{Renyi_new,WWY} with parameters $\alpha>1$. Following then Csisz\'ar's approach, we gave a direct operational interpretation of these R\'enyi relative entropies as generalized cutoff rates.
Our results show that, at least in the context of hypothesis testing, the operationally relevant quantum generalization of R\'enyi's $\alpha$-relative entropies for $\alpha>1$ are
given by $D_{\alpha}\nw$. On the other hand, previous results \cite{ANSzV,Hayashi,MH,Nagaoka} show that for $\alpha<1$, the operationally relevant quantum generalization is the traditional notion
$D_{\alpha}\old$.

Our proof for the optimality of the converse Hoeffding divergence for the strong converse rate follows immediately from the monotonicity of $D_{\alpha}\nw,\,\alpha>1$,
under measurements; this proof technique goes back to Nagaoka's proof for the strong converse \cite{Nagaoka2}.
We proved the achievability of the converse Hoeffding divergence for the strong converse rate by
showing that the quantum Neyman-Pearson tests (or suitable modifications for large $r$) achieve it for a suitably chosen trade-off parameter.
The proof uses the pinching technique developed by Hayashi \cite{H:pinching,H:text}, classical large deviation theory, and,
for \eqref{eq:58}, the asymptotic attainability of the new R\'enyi relative entropies by pinching.
An alternative proof for the achievability of the converse Hoeffding divergence can be obtained by combining the pinching 
technique with the G\"artner-Ellis theorem; this approach can be used also for the hypothesis testing problem of various for non-i.i.d.~states \cite{MO}.

\appendix
\section{Monotonicity and attainability properties of the R\'enyi divergences}
\label{sec:mon}

For a general quantum divergence $D$ (i.e., a function on pairs of density operators), one can consider various monotonicity and attainability properties.
By a monotonicity property we mean that for every $\rho,\sigma\in\B(\hil)_+$ and every 
$\map:\,\B(\hil)\to\B(\kil)$ belonging to a certain class of maps,
\begin{align}\label{mon def}
D(\map(\rho)\|\map(\sigma))\le D(\rho\|\sigma).
\end{align}
Here we will consider the monotonicity properties
MON, SMON, EPPMON, MMON and PMON, where in each case, the map $\map$ in \eqref{mon def}
is a trace-preserving positive linear map, with the following additional properties:
\vspace{.4cm}

\noindent MON:\phantom{EPI} \ds \begin{minipage}[t]{18cm}
$\map$ is completely positive.
\end{minipage}
\vspace{.01cm}

\noindent SMON:\phantom{EP} \ds \begin{minipage}[t]{14cm}
$\map$ is a stochastic map in the sense of \cite{HMPB}, i.e., it is the convex combination of two trace-preserving maps
$\map_1$ and $\map_2$, such that the adjoint (w.r.t.~the Hilbert-Schmidt inner product) of $\map_1$ is a Schwarz map, and the adjoint of $\map_2$ is a Schwarz map composed with the transposition in some basis.
\end{minipage}
\vspace{.01cm}

\noindent EPPMON: \ds \begin{minipage}[t]{14cm}
$\map$ is such that every tensor power $\map^{\otimes n}$ is positive, $n\in\bN$.
\end{minipage}
\vspace{.01cm}

\noindent MMON:\phantom{EI} \ds \begin{minipage}[t]{14cm}
$\map$ is a measurement, i.e., all operators in $\map(\B(\hil))$ commute with each other.
\end{minipage}
\vspace{.01cm}

\noindent PMON:\phantom{EP} \ds \begin{minipage}[t]{14cm}
$\map$ is the pinching with respect to the reference state $\sigma$.
\end{minipage}
\vspace{.01cm}

\noindent The following implications are obvious:

\begin{align*}
\begin{array}{ccccc}
\text{SMON} & & & & \\
\Downarrow & & & & \\
\text{MON} & \imp & \text{MMON} & \imp & \text{PMON}  \\
\Uparrow & & & & \\
\text{EPPMON} & & & & 
\end{array}
\end{align*}

By an asymptotic attainability property we mean that 
for every $\rho,\sigma\in\B(\hil)_+$, there exists a sequence of maps $\map_n:\,\B(\hil^{\otimes n})\to\B(\kil_n),\,n\in\bN$,
with each $\map_n$ belonging to some class further specified below, such that 
\begin{align*}
D(\rho\|\sigma)=\lim_{n\to+\infty}\frac{1}{n}D(\map_n(\rho^{\otimes n})\|\map_n(\sigma^{\otimes n})).
\end{align*}
Here we will consider
\vspace{.4cm}

\noindent AAM:\phantom{EI} \ds \begin{minipage}[t]{14cm}
(asymptotic attainability by measurements) \ds
Every $\map_n$ is a measurement.
\end{minipage}
\vspace{.01cm}

\noindent AAP:\phantom{EP} \ds \begin{minipage}[t]{14cm}
(asymptotic attainability by pinching) \ds
Every $\map_n$ is the pinching with respect to the reference state $\sigma^{\otimes n}$.
\end{minipage}
\vspace{.4cm}

\noindent The following implication is obvious:
\begin{align}\label{AAP->AAM}
\text{AAP}\imp\text{AAM}.
\end{align}
Furthermore, we say that $D$ satisfies AAMmax if
\begin{align*}
D(\rho\|\sigma)=\lim_{n\to+\infty}\frac{1}{n}\max_{\map_n\mathrm{ measurement}}D(\map_n(\rho^{\otimes n})\|\map_n(\sigma^{\otimes n})).
\end{align*}
We have
\begin{align}\label{MMON+AAM=EPPMON}
\text{MMON+AAM}\imp\text{AAMmax}\imp\text{EPPMON},
\end{align}
where the first implication is straightforward to verify, and the second one follows the same way as in 
Corollary \ref{cor:F monotonicity}.

The following table summarizes the monotonicity and attainability properties of the old and the new R\'enyi relative entropies
(NK stands for ``Not Known''):
\vspace{.4cm}

{\renewcommand{\arraystretch}{1.5}
\renewcommand{\tabcolsep}{0.2cm}
\begin{tabular}{|l|l|c|c|c|c|c|}
\hline
                &  &  \ds $(0,1/2)$ \ds & \ds $[1/2,1)$ \ds & \ds $(1,2]$ \ds & $(2,+\infty)$ \\
\hline
SMON & $D_{\alpha}\old$  & \multicolumn{3}{c|}{YES$^1$}                 &  NO$^2$  \\  \cline{2-6}
    & $D_{\alpha}\nw$   & NO$^3$ & \multicolumn{3}{c|}{NK}  \\ 
\hline
EPPMON & $D_{\alpha}\old$  & \multicolumn{2}{c|}{YES$^1$} & NK &       NO$^2$  \\  \cline{2-6}
    & $D_{\alpha}\nw$   & NO$^3$ &  \multicolumn{3}{c|}{YES$^4$}  \\ 
\hline
MON & $D_{\alpha}\old$  & \multicolumn{3}{c|}{YES$^1$}                 &  NO$^2$  \\  \cline{2-6}
    & $D_{\alpha}\nw$   & NO$^3$ & \multicolumn{3}{c|}{YES$^4$}  \\ 
\hline
MMON & $D_{\alpha}\old$  & \multicolumn{4}{c|}{YES$^1$}                  \\  \cline{2-6}
    & $D_{\alpha}\nw$   & NK & \multicolumn{3}{c|}{YES$^4$}  \\ 
\hline
PMON & $D_{\alpha}\old$  & \multicolumn{4}{c|}{YES$^1$}                  \\  \cline{2-6}
    & $D_{\alpha}\nw$   & \multicolumn{4}{c|}{YES$^4$}  \\ 
\hline    
AAP & $D_{\alpha}\old$  & \multicolumn{4}{c|}{NO$^5$}                  \\  \cline{2-6}
    & $D_{\alpha}\nw$   & \multicolumn{4}{c|}{YES$^4$}  \\ 
\hline    
AAM & $D_{\alpha}\old$  & NK & \multicolumn{3}{c|}{NO$^5$}                  \\  \cline{2-6}
    & $D_{\alpha}\nw$   & \multicolumn{4}{c|}{YES$^4$}  \\ 
\hline    
\end{tabular}}
\vspace{.4cm}

\noindent $^1$: Monotonicity of $D_{\alpha}\old$ for $\alpha\in[0,2]$ under $2$-positive maps has been proved in \cite{Petz},
and has been extended to stochastic maps in \cite{HMPB}. MMON and PMON for  $\alpha\in[0,2]$ are immediate consequences,
and for $\alpha>2$ they have been proved by a different method in \cite[Section 3.7]{H:text}.
EPPMON follows from the operational interpretation of $D_{\alpha}\old$ for $\alpha\in(0,1)$ in the context of the Hoeffding bound; see, e.g., \cite{Nagaoka}.
\medskip

\noindent $^2$: Failure of MON for $\alpha>2$ was pointed out in \cite[page 7]{Renyi_new}. 
One can easily see that MON is equivalent to joint convexity for the core quantities of the old R\'enyi divergences,
$Q_{\alpha}(\rho\|\sigma):=\Tr\rho^{\alpha}\sigma^{1-\alpha}$; see, e.g., \cite{Petz}. An easy argument \cite{HTp}, omitted in 
\cite{Renyi_new}, shows that even convexity of $Q_{\alpha}$ in its first argument implies the operator convexity of 
the power function $\bR_+\ni x\mapsto x^{\alpha}$. Since the latter is not true for $\alpha>2$ (see, e.g., \cite[Exercise V.2.11]{Bhatia}), MON cannot hold for $D_{\alpha}\old,\,\alpha>2$, from which the failure of SMON and EPPMON for the same range of 
$\alpha$ are obvious. 
\medskip

\noindent$^3$: MON for $D_{\alpha}\nw$ is also equivalent to joint convexity, the failure of which for $\alpha<1/2$ has been 
confirmed by numerical examples according to \cite{Renyi_new}. Failure of MON obviously yields failure of 
SMON and EPPMON.
\medskip

\noindent $^4$: 
MON for $D_{\alpha}\nw$ have been proved by various methods, applicable to different parameter ranges, in
\cite{Beigi,FL,Renyi_new,WWY}. These approaches either prove monotonicity directly, or through joint convexity, and rely on 
techniques from matrix analysis or functional analysis. 

In this paper we followed a different approach, starting from PMON, that has been proved for all parameter values
$\alpha\ge 0$ in \cite{Renyi_new}. We then proved, for $\alpha>1$, MMON in Lemma \ref{mono:measurement} and AAP in Theorem \ref{thm:attainability}, which in turn yield AAM and the stronger monotonicity property EPPMON, according to \eqref{AAP->AAM} 
and \eqref{MMON+AAM=EPPMON}; see also 
Corollary \ref{Renyi monotonicity} and Remark \ref{rem:EPPMON}.

AAP for $\alpha\in[0,1)$ has been proved very recently in \cite{HT}. It is not clear whether MMON and thus EEPMON for $\alpha
\in[1/2,1)$ can be obtained from it the same way as for $\alpha>1$ in the present paper. However, when combined with MON for 
$\alpha\in[1/2,1)$, derived
by other methods as mentioned above, it implies AAM and thus EPPMON for $\alpha\in[1/2,1)$, according to \eqref{AAP->AAM} 
and \eqref{MMON+AAM=EPPMON}.
\medskip

\noindent $^5$: 
For commuting states the old and the new R\'enyi relative entropies coincide, whereas 
for non-commuting states the inequality in \eqref{ALT} is strict according to \cite{Hiai ALT}.
Thus, AAP for $D_{\alpha}\nw$ implies that AAP cannot hold for $D_{\alpha}\old$, for any fixed value
$\alpha\in(0,+\infty)\setminus\{1\}$.
For $\alpha\ge 1/2$, AAM+MMON yields AAMmax according to \eqref{MMON+AAM=EPPMON}, and hence
\begin{align*}
\lim_{n\to+\infty}\frac{1}{n}\max_{\map_n\mathrm{ measurement}}D_{\alpha}\old(\map_n(\rho^{\otimes n})\|\map_n(\sigma^{\otimes n}))=D_{\alpha}\nw<D_{\alpha}\old
\end{align*}
whenever $\rho$ and $\sigma$ don't commute, showing that AAM fails for $D_{\alpha}\old,\,\alpha\ge 1/2$.
\medskip

\begin{remark}
In Corollary \ref{cor:EPPMON} we presented an approach to obtain EPPMON from the 
operational representation in Corollary \ref{cor:cutoff}.
However, to obtain Corollary \ref{cor:cutoff}, we used MMON (to prove Lemma \ref{lemma:converse rate lower bound}) and 
AAP (for \eqref{eq:58}), from which properties EEPMON is immediate, as we have seen above. It is an interesting open question 
whether the cutoff rate representation, or Theorem \ref{thm:exponent}, can be obtained without the use of monotonicity and achievability 
properties, thus providing a fully operational proof for the monotonicity of the new R\'enyi divergences for $\alpha>1$. We 
remark that such a fully operational proof for $D_{\alpha}\old,\,\alpha\in(0,1)$, follows from the Hoeffding bound theorem, as it 
was pointed out in \cite{Nagaoka}.
\end{remark}

\begin{remark}
For $\alpha=1$, the old and the new R\'enyi relative entropies yield the same limit $D_1$, Umegaki's relative emtropy. 
SMON and EPPMON for $D_{\alpha}\old$ yield immediately the same properties for $D_1$ by taking the limit $\alpha\to 1$.
AAP has been shown in \cite{HP}, and it was the key technical tool to prove the direct part of the quantum Stein's lemma 
\cite{HP}, and various generalizations of it \cite{BS-Sch,Shannon-McMillan,BDKSSSz}.  
From these, the rest of the properties, MON, MMON, PMON, AAM and AMMmax, follow immediately, as we have seen before.
\end{remark}
\medskip

The above properties show that the new R\'enyi relative entropies provide the smallest possible quantum extension
of the classical R\'enyi relative entropies, under very mild conditions.
\begin{prop}
For a fixed $\alpha\ge 0$, let $\what D_{\alpha}$ be a function on pairs of quantum states on the same Hilbert space, with the following properties:
\begin{enumerate}
\item
$\what D_{\alpha}$ coincides with the classical R\'enyi relative entropy $D_{\alpha}$ on commuting states;
\item
$\what D_{\alpha}$ is additive, i.e., for every $\rho,\sigma$ and every $n\in\bN$,
$\what D_{\alpha}(\rho^{\otimes n}\|\sigma^{\otimes n})=n\what D_{\alpha}(\rho\|\sigma)$;
\item
$\what D_{\alpha}$ satisfies PMON.
\end{enumerate}
Then $D_{\alpha}\nw\le \what D_{\alpha}$.
In particular, $D_{\alpha}\nw\le D_{\alpha}\old$ for every $\alpha\in[0,+\infty]\setminus\{1\}$.
\end{prop}
\begin{proof}
Let $\rho$ and $\sigma$ be fixed. By assumption, we have
\begin{align*}
D_{\alpha}(\E_{\sigma^{\otimes n}}(\rho^{\otimes n})\|\sigma^{\otimes n})
=
\what D_{\alpha}(\E_{\sigma^{\otimes n}}(\rho^{\otimes n})\|\sigma^{\otimes n})
\le
\what D_{\alpha}(\rho^{\otimes n}\|\sigma^{\otimes n})
=
n\what D_{\alpha}(\rho\|\sigma).
\end{align*}
Using that $D_{\alpha}\nw$ satisfies AAP, we get
\begin{align*}
D_{\alpha}\nw(\rho\|\sigma)=
\lim_{n\to+\infty}\frac{1}{n}D_{\alpha}(\E_{\sigma^{\otimes n}}(\rho^{\otimes n})\|\sigma^{\otimes n})
\le
\what D_{\alpha}(\rho\|\sigma).
\end{align*}
\end{proof}
\vspace{.4cm}

\noindent\textbf{Sufficiency and single-shot attainability}
\vspace{.4cm}

Instead of the asymptotic attainability properties studied above, one can also consider single-shot attainability.
Here we will be interested in attainability by measurements (AM), which is satisfied by a quantum divergence $D$ if for every 
pair of states $\rho,\sigma$, there exists a measurement $\F$ such that 
$D(\F(\rho)\|\F(\sigma))=D(\rho\|\sigma)$. It is easy to see that 
\begin{align}\label{AM+MMON=posMON}
\text{AM+MMON}\imp\text{monotonicity under trace-preserving positive maps},
\end{align}
a very strong monotonicity property. It is clear that $D_{\alpha}\old$ cannot satisfy AM for any 
$\alpha\in(0,+\infty)\setminus\{1\}$, due to the strict inequality in 
\eqref{ALT} for non-commuting states. It is an open question whether AAM for $D_{\alpha}\nw$ can be strengthened to AM in 
general. However, we have the following special cases:
\begin{lemma}\label{A}
$D_{1/2}\nw$ and $D_{+\infty}\nw=D_{\max}$ satisfy AM.
\end{lemma}
\begin{proof}
Note that $D_{1/2}\nw=-2\log F$, where $F$ is Uhlmann's fidelity \cite{Uhlmann2}.
Since the fidelity is known to be attainable by measurements (see, e.g., \cite[Chapter 9]{NC}), the assertion follows for
$D_{1/2}\nw$.

If $\rho,\sigma\in\B(\hil)_+$ are such that $\supp\rho\le\supp\sigma$ then one can use the duality of linear programming to write the max-relative entropy of $\rho$ and $\sigma$ as \cite{BFS,marco_thesis,Winter2}
\begin{align*}
D_{\max}(\rho\|\sigma)&=
\max\{\log\Tr M\rho\,:\,0\le M,\,\Tr M\sigma=1\}\\
&=
\max\left\{\log\frac{\Tr M\rho}{\Tr M\sigma}\,:\,0\le M\le I\right\}\\
&=
\max\left\{\max_{x\in\X}\left\{\log\frac{\Tr M_x\rho}{\Tr M_x\sigma}\right\}\,:\,\{M_x\}_{x\in\X}\text{ POVM}\right\}\\
&=
\max\left\{D_{\max}(\{\Tr\rho M_x\}_{x\in\X}\|\{\Tr\sigma M_x\}_{x\in\X})\,:\,\{M_x\}_{x\in\X}\text{ POVM}\right\}.
\end{align*}
The equality between the first and the last expression above holds trivially when $\supp\rho\le\supp\sigma$ is not satisfied.
\end{proof}

It is well-known that the fidelity is monotone non-decreasing, or equivalently, $D_{1/2}\nw$ is monotone non-increasing, under 
CPTP maps. Combining this with Lemma \ref{A}, we get the following stronger monotonicity property:
\begin{corollary}
The fidelity is monotone non-decreasing, or equivalently, $D_{1/2}\nw$ is monotone non-increasing, under trace-preserving positive maps.
\end{corollary}
\begin{proof}
Monotonicity under CPTP maps implies MMON, and thus the assertion is immediate from 
Lemma \ref{A} and \eqref{AM+MMON=posMON}.
\end{proof}

\begin{remark}
Monotonicity of $D_{\max}$ under trace-preserving positive maps is trivial from its definition \eqref{zero NP}.
\end{remark}

\begin{remark}
It is easy to see that for fixed states, the classical R\'enyi relative entropies are monotone increasing in the parameter $\alpha$.
Lemma \ref{A} thus yields that 
\begin{align*}
D_{\max}(\rho\|\sigma)&=
\max_{\alpha\in[0,+\infty]}\max\left\{D_{\alpha}\bz\{\Tr M_i\rho\}\|\{\Tr M_i\sigma\}\jz\,:\,\{M_i\}\text{ POVM}\right\},
\end{align*}
i.e., the max-relative entropy of $\rho$ and $\sigma$ is the largest R\'enyi $\alpha$-relative entropy
of the classical distributions that can be obtained from $\rho$ and $\sigma$ after performing a measurement. 
\end{remark}
\smallskip

We say that a quantum divergence $D$ satisfies the sufficiency property (S) if the following holds: For every states
$\rho,\sigma\in\S(\hil)$, and CPTP map $\map:\,B(\hil)\to\B(\kil)$, 
\begin{align}\label{sat}
D(\map(\rho)\|\map(\sigma))=D(\rho\|\sigma)
\end{align}
implies the 
existence of a CPTP map $\mapp:\,\B(\kil)\to\B(\hil)$ such that 
\begin{align}\label{rev}
\mapp(\map(\rho))=\rho\ds\ds\ds\text{and}\ds\ds\ds \mapp(\map(\sigma))=\sigma.
\end{align} 
Obviously, if $D$ is monotone under CPTP maps then \eqref{rev} implies \eqref{sat}. Thus, for a monotone divergence, sufficiency 
means that the monotonicity inequality is strict in the sense that it can only be saturated in a trivial way.

The old R\'enyi relative entropies $D_{\alpha}\old$ satisfy MON for every $\alpha\in[0,2]$, and they are 
known to have the sufficiency property for every parameter value in this interval, except for its endpoints $0$ and $2$;
see \cite{HMPB,JP,JP2,Petz2,Petz3}.
Failure of (S) for $\alpha=0$ is trivial to see, and for $\alpha=2$ it follows from a counterexample given in
\cite[Example 2.2]{JPP} and \cite[Section 5]{HMPB}. 

Sufficiency for the new R\'enyi relative entropies is an open question for every parameter value, except at the endpoints
of the monotonicity interval $[1/2,+\infty]$. Below we show that, similarly to the case of the old R\'enyi relative entropies,
sufficiency fails at these points.

The following lemma is due to Petz \cite[Lemma 4.1]{Petz4}.
\begin{lemma}\label{lemma:Petz}
Let $\rho,\sigma$ be states and $\{M_x\}_{x\in\X}$ be a measurement such that
\begin{equation}\label{D half equality}
D_{1/2}\old\bz\{\Tr\rho M_x\}_{x\in\X}\|\{\Tr\sigma M_x\}_{x\in\X}\jz=
D_{1/2}\old(\rho\|\sigma).
\end{equation}
Then $\rho$ and $\sigma$ commute.
\end{lemma}

\begin{corollary}
No quantum divergence can satisfy (A)+(S). In particular, $D_{1/2}\nw$ and $D_{\infty}\nw$ do not satisfy (S).
\end{corollary}
\begin{proof}
Assume that $D$ satisfies (A) and (S), and let $\rho,\sigma$ be non-commuting states.
By (A), there exists a POVM $\{M_x\}_{x\in\X}$ such that
$D(\rho\|\sigma)=D\bz\{\Tr\rho M_x\}_{x\in\X}\|\{\Tr\sigma M_x\}_{x\in\X}\jz$. By (S), there exists a CPTP map $\Psi$ such that
$\Psi(\{\Tr\rho M_x\}_{x\in\X})=\rho$ and $\Psi(\{\Tr\sigma M_x\}_{x\in\X})=\sigma$. By the monotonicity of
$D_{1/2}\old$, we have \eqref{D half equality}, and by Lemma \ref{lemma:Petz}, $\rho$ and $\sigma$
commute, which is a contradiction.

The assertion about $D_{1/2}\nw$ and $D_{\infty}\nw$ follows as a special case, due to Lemma \ref{A}.
\end{proof}

\section*{Acknowledgments}

The authors would like to thank Dr.~Gen Kimura for his hospitality, and MM would further like to thank Prof.~Fumio Hiai, Dr.~Hiromichi Ohno and Prof.~Takashi Sano for their hospitality during his visit in Japan.
MM acknowledges support by the European
Commission (Marie Curie Fellowship ``QUANTSTAT'') and by the European Research Council (Advanced Grant``IRQUAT''). Part of this work was done when MM was a Marie Curie research fellow at the School of Mathematics, University of Bristol.
TO was partially supported by MEXT Grant-in-Aid
(A) No.~20686026 ``Project on Multi-user Quantum Network''.
The authors are grateful to Nilanjana Datta, Masahito Hayashi, Ke Li, Marco Tomamichel and Andreas Winter for comments 
and discussions, and to an 
anonymous referee for his/her comments. The authors are also grateful to Masahito Hayashi and Marco Tomamichel for sharing with them the manuscript \cite{HT} before publication.

\end{document}